\documentclass[runningheads]{llncs}
\usepackage[utf8]{inputenc}

\usepackage{array}
\usepackage{xcolor}
\usepackage{todonotes}
\usepackage{graphicx}
\usepackage{subfigure}
\usepackage{pgf,tikz}
\usepackage{mathrsfs}
\usepackage{amsmath}
\usetikzlibrary{arrows}

\usepackage{lineno}

\let\doendproof\endproof
\renewcommand\endproof{~\hfill$\qed$\doendproof}

\newcommand{\remove}[1]{{}}
\newcolumntype{P}[1]{>{\centering\arraybackslash}p{#1}}

\title{Improved Bounds for Drawing Trees on \\ Fixed Points with L-shaped Edges }
\author{Therese Biedl\inst{1} \and 
Timothy M. Chan\inst{2} \and
Martin Derka\inst{1}  \and \\
Kshitij Jain\inst{1} \and 
Anna Lubiw\inst{1}
}
\institute{University of Waterloo, Waterloo, ON, N2L 3G1, Canada\\
\email{\{biedl,mderka,k22jain,alubiw\}@uwaterloo.ca} % ADD YOUR EMAIL
\and
University of Illinois at Urbana-Champaign,
\email{tmc@illinois.edu}
}
\authorrunning{T. Biedl, T. M. Chan, M. Derka, K. Jain, A. Lubiw}
\date{May 2017}

\begin{document}

% switch here from GD to arxiv
\newif\ifarxiv 
\arxivtrue
%\arxivfalse
% tell me where the figures are
%\graphicspath{{figures/}}

\ifarxiv
%\linenumbers
\fi

\maketitle

\begin{abstract}
Let $T$ be an $n$-node tree of maximum degree 4, 
and let $P$ be a set of $n$ points in the plane with no two points on the same horizontal or vertical line.  
It is an open question whether $T$ always has a planar drawing on $P$ such that each edge is drawn as an orthogonal path with one bend (an ``L-shaped'' edge). 
By giving new methods for drawing trees, we improve the bounds on the size of the point set $P$ for which such drawings are possible to: $O(n^{1.55})$ for maximum degree 4 trees; $O(n^{1.22})$ for maximum degree 3 (binary) trees; and $O(n^{1.142})$ for perfect binary trees.

Drawing ordered trees with L-shaped edges is harder---we give an example that cannot be done and a bound of $O(n \log n)$ points for L-shaped drawings of ordered caterpillars, which 
contrasts with
% from Anna: one upper bound contrasts with the other upper bound.  I don't see why the referee objects.
%complements 
the known linear bound for unordered caterpillars. 
\end{abstract}

\section{Introduction}

The problem of drawing a planar graph so that its vertices are
restricted to a specified set of points in the plane has been well-studied, both from the perspective of algorithms and from the perspective of bounding the size of the point set and/or the number of bends needed to draw the edges.  
Throughout this paper we consider the version of the problem where the points are unlabelled, i.e.,~we may choose to place any vertex at any point.

One basic result is that every planar $n$-vertex graph has a planar drawing on any set of $n$ points, even with the limitation of at most 2 bends per edge~\cite{Kaufmann-Wiese}. 
If every edge must be drawn as a straight-line segment then any $n$ points in general position still suffice for drawing trees~\cite{Bose-trees} and outerplanar graphs~\cite{Bose-outerplanar} but 
this result does not extend to any non-outerplanar graph~\cite{Gritzmann}, 
and the decision version of the problem becomes NP-complete~\cite{Cabello}.
Since $n$ points do not always suffice, the 
next natural question is: How large must a \emph{universal} point set be, and how many points are needed for \emph{any} point set to be universal?  
An upper bound of $O(n^2)$ follows from the 1990 algorithms that draw planar graphs on an $O(n) \times O(n)$ grid~\cite{dFPP,Sch}, but the best known lower bound, dating from 1989, is  
$c \cdot n$ for some $c>1$~\cite{universal-lower-bound}.  

Although orthogonal graph drawing has been studied for a long time, analogous questions of universal point sets for orthogonal drawings have only been explored recently, beginning with  Katz et al.~\cite{Katz-manhattan} in 2010.
Since at most 4 edges can be incident to a vertex in an orthogonal drawing, attention is restricted to graphs of maximum degree 4.
Throughout the paper we will restrict attention to 
point sets  in ``general orthogonal position'' meaning that no two points share the same $x$- or $y$-coordinate.
We study the simplest type of orthogonal drawings 
where every edge must be drawn as an orthogonal path of two segments.
Such a path is called an ``L-shaped edge'' and these drawings are called ``planar L-shaped drawings''.
Observe that any planar L-shaped drawing lives in the grid formed by the $n$ horizontal and $n$ vertical lines through the points. 

Di Giacomo et al.~\cite{DiGiacomo} introduced the problem of planar L-shaped drawings and showed that any tree of maximum degree 4 has a planar L-shaped drawing on any set of $n^2 -2n + 2$ points (in general orthogonal position, as will be assumed henceforth).  Aichholzer et al.~\cite{Aichholzer} improved the bound to $O(n^c)$ with $c = \log_2 3 \approx 1.585$.
It is an open question whether $n$ points always suffice. 
Surprisingly, nothing better is known for trees of maximum degree 3.
%---the upper bound is the same, and it is open whether $n$ points always suffice.

The largest subclass of trees for which $n$ points are known to suffice is the class of caterpillars of maximum degree 3~\cite{DiGiacomo}.  A
\emph{caterpillar} is a tree such that deleting the leaves gives a path, called the \emph{spine}.
For caterpillars of maximum degree 4 with $n$ nodes, any point set of 
size $3n - 2$ permits a planar L-shaped drawing~\cite{DiGiacomo}, and the factor was improved to $5/3$ by Scheucher~\cite{Scheucher-thesis}.

\subsection{Our Results}

We give improved bounds as shown in Table~\ref{table:results}.  A tree of max degree 3 (or 4) is \emph{perfect} if it is a rooted binary tree (or ternary tree, respectively) in which all leaves are at the same height and all non-leaf nodes have 2 (or 3, respectively) children.
Our bounds are achieved by constructing the drawings recursively and analyzing the resulting recurrence relations, 
which is the same approach used previously by Aichholzer et al.~\cite{Aichholzer}.    
Our improvements come from more elaborate drawing methods. 
Results on maximum degree 3 trees are in Section~\ref{sec:degree-3} and results on maximum degree 4 trees are in Section~\ref{sec:degree-4}.

\renewcommand{\arraystretch}{1.2}%{2}
\begin{table}[h!]
  \centering
  \caption{Previous and new bounds on the number of points sufficient for planar L-shaped drawings of any tree of $n$ nodes.  The previous bounds all come from Aichholzer et al.~\cite{Aichholzer}.}
  \label{table:results}
    \begin{tabular}{|l|c|c|}
  \hline
    &  previous  & new \\
    \hline \hline
    deg 3 perfect\ \ \ \ & $n^{1.585}$ & \ \ $n^{1.142}$\ \ \\ 
    deg 3 general\ & $n^{1.585}$ & $n^{1.22}$\\
    deg 4 perfect\ & $n^{1.465}$\footnotemark{} & \\
    deg 4 general\ & $n^{1.585}$ &  $n^{1.55}$ \\
    \hline
  \end{tabular}
\remove{%%%%%%%%
  \caption{Our bounds on the number of points sufficient for planar L-shaped drawings of any tree of $n$ nodes.  The previous best bounds are indicated in parentheses and all come from Aichholzer et al.~\cite{Aichholzer}.}
  \begin{tabular}{|l|c|c|}
  \hline
    & \ perfect\  & \ general\ \\
    \hline \hline
    max deg 3\  & $n^{1.142}$ ($n^{1.585}$) & $n^{1.22}$ ($n^{1.585}$)\\
    \hline
    max deg 4\  & ($n^{1.465}$)\footnotemark{} & $n^{1.55}$ ($n^{1.585}$)\\
    \hline
  \end{tabular}
  }%%%%%%%%%%%%%
\end{table}
\renewcommand{\arraystretch}{1}

\footnotetext{The bound of $n^{1.465}$ is not explicit in~\cite{Aichholzer} but will be explained in Section~\ref{sec:degree-4}.}

\remove{
\begin{center}
\begin{tabular}{|r|c|c|}
\hline
& perfect trees & arbitrary trees \\ \hline \hline 
binary trees & $n^{1.142}$ & $n^{1.22}$ \\ \hline % 1.142 is Timothy's e-mail Dec 19, it had some hand-drawn images attached, we do not use h here
ternary trees & $n^{1.26}$ & $n^{1.55}$ \\ \hline
\end{tabular}
\end{center}
}

We also consider the case of \emph{ordered} trees where the cyclic order of edges incident to each vertex is specified.  We give an example of an $n$-node ordered tree (in fact, a caterpillar) and a set of $n$ points such that the tree has no L-shaped planar drawing on the point set. 
We also give a positive result about drawing some ordered caterpillars on $O(n \log n)$ points. The caterpillars that can be drawn on such $O(n \log n)$ points include our example that cannot be drawn on a given set of $n$ points. These results are in Section~\ref{sec:caterpillars}.

\subsection{Further Background}

Katz et al.~\cite{Katz-manhattan} introduced the problem of drawing a planar graph on a specified set of points in the plane so that each edge is an \emph{orthogeodesic path}, i.e.~a path of horizontal and vertical segments whose length is equal to the $L_1$ distance between the endpoints of the path. 
%This is analogous to planar straight-line drawings in the non-orthogonal world. 
They showed that the problem of deciding whether an $n$-vertex planar graph has a planar orthogeodesic drawing on a given set of $n$ points is NP-complete.
Subsequently, Di Giacomo et al.~\cite{DiGiacomo} showed that any $n$-node tree of maximum degree 4 has an orthogeodesic drawing on any set of $n$ points where the drawing is restricted to the $2n \times 2n$ grid that consists of the ``basic'' horizontal and vertical lines through the points, plus one extra line between each two consecutive parallel basic lines.  If the drawing is restricted to the basic grid, their bounds were $4n-3$ points for degree-4 trees, and $3n/2$ points for degree-3 trees.  These bounds were improved by Scheucher~\cite{Scheucher-thesis} and then by B\'ar\'any et al.~\cite{Barany}. 

%\ournote{check if Scheucher has more on caterpillars. Martin: He does. Whole chapter 5. He does not work with straight-through drawings, and his drawings are not straight-through. In Lemma 20, he shows that two drawn caterpillars can, if drawn nicely (with some specific rays on the ends free), then they can be chained (this chaining does not result in straight-through). Then he draws caterpillars one spine vertex at the time and reaches a bound $4/3 \cdot n + O(1)$ for othogeodesic drawings and 5/3 factor for L-drawings (with the proof being computer-assisted). Apparently, Giacomo et al. had bound $1.5n + O(1)$ for orthogeodesic. Scheucher conjectures that 4n/3 can be improved.}

%%%%%%%%%%%%%%%%%%%%%%%%%%%%%%%%%%%%%%%%%%%%%%%%%%%%
\section{Ordered Trees---the Case of Caterpillars}
\label{sec:caterpillars}

All previous work has assumed that trees are unordered, i.e.,~that we may freely choose the cyclic order of edges incident to a vertex. 
In this section we show that ordered trees on $n$ vertices do not always have planar L-shaped drawings on $n$ points.  
Our counterexample is a \emph{top-view caterpillar}, i.e., a caterpillar such that the two leaves adjacent to each vertex lie on opposite sides of the spine.
%
%with its leaves alternating the sides of its spine.
The main result in this section is that 
every top-view caterpillar has a planar L-shaped drawing on $c n \log n$ points for some $c>0$. 
%We conjecture that the correct bound is $2n$ points.

First the counterexample.  We prove the following in 
\ifarxiv
Appendix~\ref{appendix:14-caterpillar}:
\else
the full version:
\fi

\begin{lemma}
The top-view caterpillar $C_{14}$ on $n=14$ nodes shown in 
Figure~\ref{fig:caterpillar}(a) cannot be drawn with L-shaped edges on the point set $P_{14}$ of size $14$ shown in 
Figure~\ref{fig:caterpillar}(c).
\label{claim:14-caterpillar}
\end{lemma}

It is conceivable that this counter-example is an isolated one---we have been unable to extend it to a family of such examples. 

%The known results for caterpillars---that $n$ points suffice for degree 3 caterpillars, and $3n - 2$ points suffice for %degree 4 caterpillars~\cite{DiGiacomo}---in fact use drawings of the caterpillars with all leaves to the same side of the %spine (``side-view'').
%\ournote{We need to double-check this!}  NO, it'd not true.  The proofs us the freedom to place leaves above or below.

\begin{figure}

\centering
\includegraphics[width=0.7\linewidth]{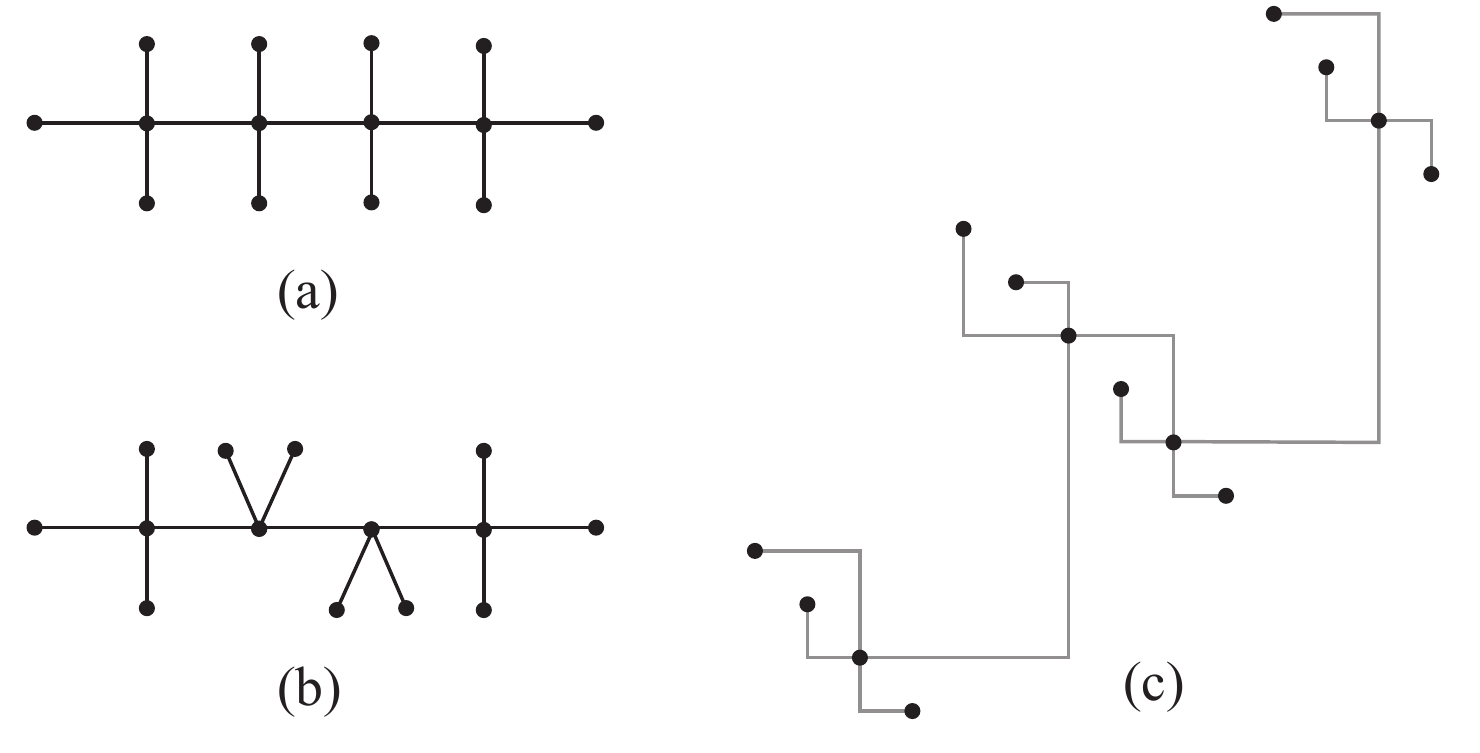}
\caption{The ordered top-view caterpillar $C_{14}$ shown in (a) does not have a planar L-shaped drawing on the point set $P_{14}$ shown in (c).  The ordering shown in (b) does.}
\label{fig:caterpillar}
\end{figure}

Next we explore the question of how many points are needed for a planar L-shaped drawing of an $n$-vertex top-view caterpillar. 
Consider the appearance of the caterpillar's spine (a path) in such a drawing.  
Each node of the spine, except for the two endpoints, must have its two incident spine edges aligned---both horizontal or both vertical. 
Define a \emph{straight-through} drawing of a path to be a planar L-shaped drawing such that the two incident edges at each vertex are aligned. 
The best bound we have for the number of points that suffice for a straight-through drawing of a path is obtained when we draw the path in a monotone fashion, i.e.~i.e. with non-decreasing x-coordinates. 

\begin{theorem} Any path of $n$ vertices has an $x$-monotone straight-through drawing on any set of at least $c \cdot n \log n$ points for some constant $c$.
\label{thm:mono-path}
\end{theorem}
\begin{proof}
We prove that if the number of points satisfies the recurrence $M(n) = 2 M(\frac{n}{2}) + 2n$ then any path of $n$ vertices has an $x$-monotone straight-through drawing on the points.  Observe that this recurrence relation solves to $M(n) \in \Theta(n \log n)$ which will complete the proof.  
Within a constant factor we can assume without loss of generality that $n$ is a power of 2. 

Order the points by $x$-coordinate. 
Recall our assumption that no two points share the same $x$- or $y$-coordinate.
By induction, the first half of the path has an $x$-monotone straight-through drawing on the first $M(\frac{n}{2})$ points. 
We add the assumption that the path starts with a horizontal segment.

Let $p$ be the second last point used.  Since $n$ is a power of 2, the path goes through $p$ on a horizontal segment.
Let $T$ be the set of points to the right of and above $p$.  Let $B$ be the set of points to the right of and below $p$.
Refer to Figure~\ref{fig:mono-path}(a).
In $T$, consider the partial order $(x_1,y_1) \prec_T (x_2,y_2)$ if $x_1 < x_2$ and $y_1 < y_2$.  Let $T'$ be the set of minimal elements in this partial order.  
Similarly, in $B$, let $B'$ be the set of elements that are minimal in the ordering $(x_1,y_1) \prec_B (x_2,y_2)$ if $x_1 < x_2$ and $y_1 > y_2$.  If $T'$ has $n$ or more points, then we can draw the whole path on $T'$ with an $x$-monotone straight-through drawing starting with a horizontal segment.  The same holds if $B'$ has $n$ or more points.  Thus we may assume that $|T'|, |B'| < n$.  
We now remove $T'$ and $B'$; let $R = (T - T') \cup (B - B')$.  Then $|R| \ge M(\frac{n}{2})$.

\begin{figure}
\label{fig:mono-path}
\centering
\includegraphics[width=0.9\linewidth]{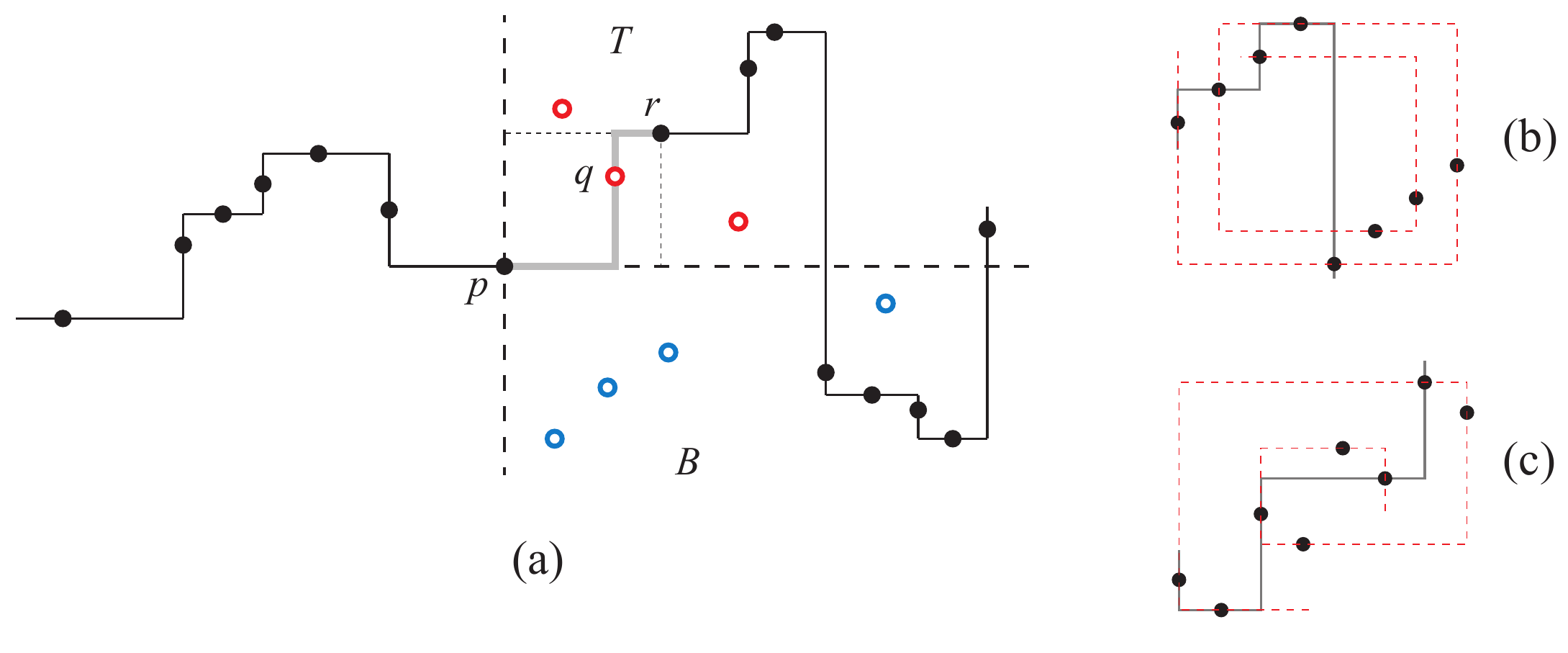}
\caption{(a) The construction for the proof of Theorem~\ref{thm:mono-path}.  The points of $T'$ and $B'$ are drawn as hollow red points above $p$ and hollow blue points below $p$, respectively.
(b-c) Examples of point sets of size $2n$ for which the maximum length of an $x$-monotone straight-through path is $n+1$.  Such paths are shown in grey. In both cases there are non-monotone planar straight-through paths of length $2n$ (dashed).
}
\end{figure}

By induction the second half of the path has an $x$-monotone straight-through drawing on the set $R$ starting with a horizontal segment.  Let $r$ be the first point used for this drawing.
Assume without loss of generality that $r$ lies in $T$.  (The other case is symmetric.)
Consider the rectangle with opposite corners at $p$ and $r$.  Since $r$ is not in $T'$, there is a point $q \in T$ inside the rectangle.  We can join the two half paths using a vertical segment through $q$ and the last vertex of the first half path is embedded at $q$.    
\end{proof}

We can extend the above result to draw the entire caterpillar (not just its spine) with the same bound on the number of points:

\begin{theorem}
Any top-view caterpillar of $n$ vertices has a planar L-shaped drawing on any set of $c \cdot n \log n$ points for some constant $c$.
\label{thm:mono-caterpillar}
\end{theorem}
\begin{proof}[outline]
Follow the above construction, but in addition to $T'$ and $B'$, also take the second and third layers. If any layer has $n$ or more points, we embed the whole caterpillar on it \cite{DiGiacomo}.  Otherwise, we remove at most a linear number of points, and embed the second half of the caterpillar by induction on the remaining points.  Then, in the rectangle between $p$ and $r$ there must be an increasing sequence of 3 points.  Use the middle one for the left-over spine-vertex $q$ and the other two for the leaves of $q$.
\end{proof}

We conjecture that $2n$ points suffice for an $x$-monotone straight-through drawing of any $n$-path.  See Figure~\ref{fig:mono-path}(b-c) for a lower bound of $2n$.
Do $n$ points suffice if the $x$-monotone condition is relaxed to planarity?
%The problem of finding the correct bound on the number of points needed for an $x$-monotone straight-through drawing of an $n$-path seems interesting in its own right.  We proved an upper bound of $c \cdot n \log n$. 
%Examples that give a lower bound of $2n$ are shown in Figure~\ref{fig:mono-path}(b-c).
%We conjecture that the right bound is $2n$.  
%Relaxing the condition of monotonicity to planarity should give a better bound, but we leave this as an open problem:  How many points suffice for planar straight-through drawings of an $n$-node path?  In particular, do $n$ points always suffice?
Finally, we mention that the problem of finding monotone straight-through paths is related to a problem about alternating runs in a sequence, as explained 
\ifarxiv
in Appendix~\ref{app:runs}.
\else
in the full paper.
\fi

%%%%%%%%%%%%%%%%%%%%%%%%%%%%%%%%%%%%%%%%%%%%%%%%%%%%
\section{Trees of Maximum Degree 3}
\label{sec:degree-3}

In this section, we prove bounds on the number of points needed for L-shaped drawings of trees with maximum
degree $3$. We treat the trees as rooted and thus, we refer
to them
%to the trees in this section 
as binary trees. 
We name the parts of the tree as shown in Figure~\ref{fig:tree-set-up}(a).  The root $r_0$ has two subtrees $T_1$ and $T_2$ of size $n_1$ and $n_2$, respectively, with $n_1 \le n_2$. 
$T_2$'s root, $r_1$, has subtrees of sizes $n_{2,1}$ and $n_{2,2}$ with $n_{2,1} \le n_{2,2}$.

\begin{figure}[htb]
\centering
\includegraphics[width=\linewidth]{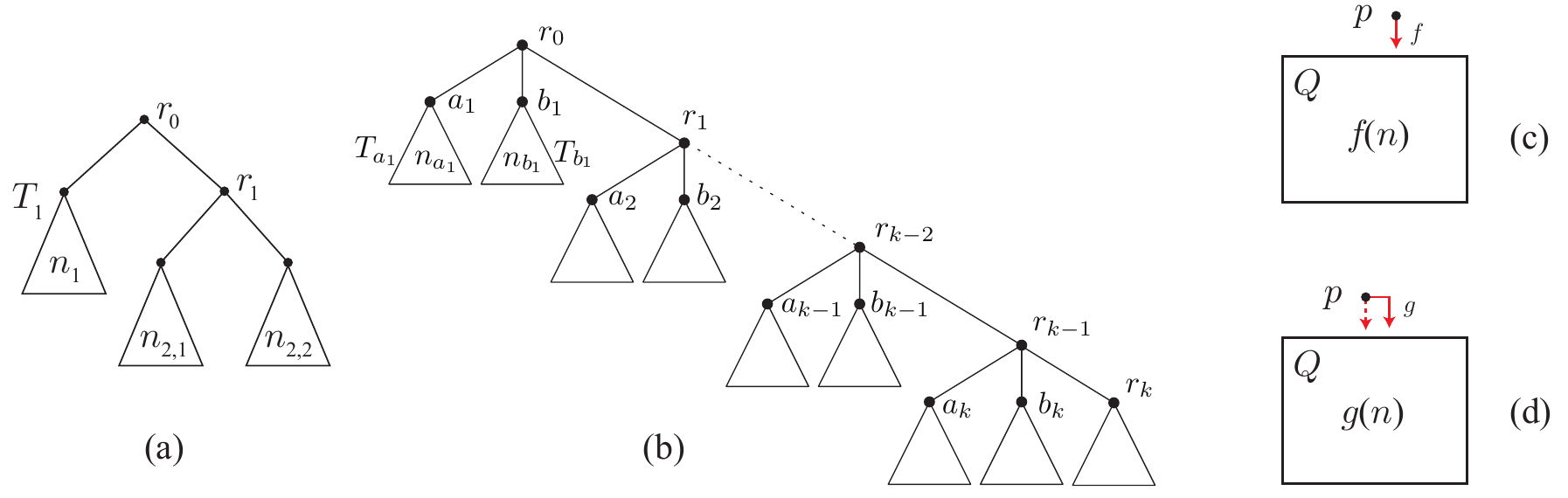}
\caption{The naming conventions for (a) binary and (b) ternary trees.  The set-up for (c) $f$-configurations and (d) $g$-configurations.}
\label{fig:tree-set-up}
\end{figure}

The main idea is to draw a tree $T$ on a set of points in a rectangle $Q$ by partitioning the rectangle into subrectangles in which we recursively draw subtrees.  This gives rise to recurrence relations  
for the number of points needed to draw trees of size $n$, which we then analyze.
We distinguish two subproblems or ``configurations.'' In each, we must draw a tree $T$ rooted at $r_0$ in a rectangle $Q$ that currently has no part of the drawing inside it.  Furthermore, the parent $p$ of $r_0$ has already been drawn, and one or two rays outgoing from $p$ have been reserved for drawing the first segment of edge $(p,r_0)$ (without hitting any previous part of the drawing). 

In the \emph{$f$-configuration} the reserved ray from $p$ goes vertically downward to $Q$. See Fig.~\ref{fig:tree-set-up}(c). 
Let $f(n)$ be the smallest number of points such that any binary tree with $n$ vertices can be
drawn in any rectangle with $f(n)-1$ points in the $f$-configuration\footnote{Beware: we will use the same notation $f(n)$ in Section~\ref{sec:degree-4} to refer to ternary trees.}. 
%\anote{Equivalently, draw a binary tree with $n$ vertices plus the parent node.}
We will give various ways of drawing trees in the $f$-configuration, each of which gives an upper bound on $f(n)$. Among these choices, the algorithm uses the one that requires the fewest points. 
%The algorithm uses of these ways the one that requires the smallest number of points.
%(The subscript `3' identifies this function as the one for degree-3 trees.)

In the \emph{$g$-configuration} we reserve a horizontal ray from $p$, that allows the L-shaped edge $(p,r_0)$ to turn downward into $Q$ at any point without hitting any previous part of the drawing.  In addition, we reserve the vertical ray downward from $p$ in case this ray enters $Q$.
See Fig.~\ref{fig:tree-set-up}(d) for the case where the horizontal ray goes to the right.
Let $g(n)$ be the smallest number of points such that any binary tree with $n$ vertices can be
drawn in any rectangle with $g(n)-1$ points in the $g$-configuration.  
%(We omit a subscript `3' since the $g$-configuration is only used for degree-3 trees.)  
Observe that $f(n) \ge g(n)$ since the $g$-configuration gives us strictly more freedom. 

We start with two easy constructions to give the flavour of our methods.

\begin{figure}[htb]
\centering
\includegraphics[width=\linewidth]{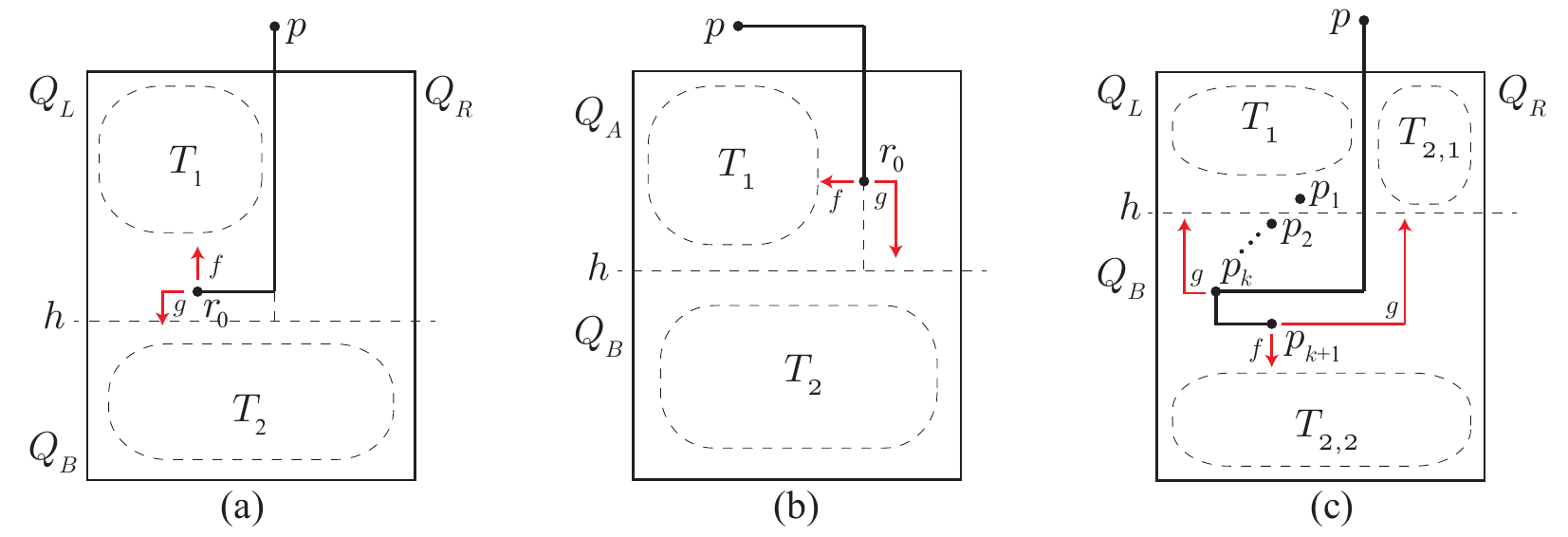}
\caption{Three methods: (a) $f$-draw-1; (b) $g$-draw; and (c) $f$-draw-2.}
\label{fig:fdraw1-gdraw}
\end{figure}

\smallskip\noindent{\bf $f$-draw-1.}  This method, 
illustrated in Figure~\ref{fig:fdraw1-gdraw}(a), applies to an $f$-configuration.  We first describe the construction and then say how many points are required.
%where $Q$ has $2f(n_{1})+g(n_{2}) - 1 $ points.  
Continue the vertical ray from $p$ downward to a horizontal half-grid line $h$ determined as follows. 
Partition $Q$ by $h$ and the ray down to $h$ into 3 parts: $Q_B$, the rectangle below $h$; $Q_L$, the upper left rectangle; and $Q_R$, the upper right rectangle. 
Choose $h$ to be the highest half-grid line such that $Q_L$ or $Q_R$ has $f(n_{1})$ points.
Without loss of generality, assume that $Q_L$ has $f(n_{1})$ points, and $Q_R$ has at most $f(n_{1})$ points.
Place $r_0$ at the bottommost point of $Q_L$. Draw the edge $(p,r_0)$ down and left.  Start a ray vertically up from $r_0$, and recursively draw $T_{1}$ in $f$-configuration (rotated $180^\circ$) in the subrectangle of $Q_L$ above $r_0$, which has $f(n_1)-1$ points.  
This leaves the leftward and downward rays free at $r_0$, so we can draw $T_{2}$ recursively in $g$-configuration in $Q_B$ so long as there are $g(n_2) - 1$ points.
The total number of points required is $2f(n_1) + g(n_2) - 1$. Recall that $f(n)$ is 1 more than the number of required points, so this proves:  
\begin{align}\tag{$f$-1}
f(n)\leq 2f(n_{1})+g(n_{2}).
\end{align}
Observe that we could have swapped $f$ and $g$ %in the above construction 
which proves:
\begin{align}\tag{$f$-1'}
f(n)\leq 2g(n_{1})+f(n_{2}).
\end{align}
The above method can be viewed as a special case of Aichholzer et al.'s method for ternary trees~\cite{Aichholzer} (see Section~\ref{sec:degree-4}).  We incorporate two new ideas to improve their result: first, they used only $f$-configurations, but we notice that one of the above two recursive subproblems is a $g$-configuration in the binary tree case, and can be solved by a better recursive algorithm; second, their method wasted all the points in $Q_R$, but 
we will give more involved constructions that allow us to use some of those points.
%if the set is large. 

%We have wasted all the points in $Q_R$. Below, we will give more involved constructions that allow us to use some of those points if the set is large.  

\smallskip\noindent{\bf $g$-draw.}
This method 
applies to a $g$-configuration where the ray from the parent node $p$ goes to the right. 
%where $Q$ has $f(n_1) + g(n_2) - 1$ points 
Partition $Q$ at the highest horizontal half-grid line such that the top rectangle $Q_A$ has $f(n_1)$ points.  
%Then the bottom rectangle $Q_B$ has $g(n_2) - 1$ points. 
We separate into two cases depending whether the rightmost point, $q$, of $Q_A$ is to the right or left of $p$.

If $q$ is to the right of $p$, place $r_0$ at $q$, and draw the edge $(p,r_0)$ right and down. 
See Figure~\ref{fig:fdraw1-gdraw}(b).  Start a ray leftward from $r_0$ and recursively draw $T_1$ in $f$-configuration in the subrectangle of $Q_A$ to the left of $q$.  Note that there are $f(n_1) - 1$ points here, which is sufficient.  The rightward and downward rays at $r_0$ are free, so we can draw $T_2$ recursively in $g$-configuration in $Q_B$ if there are $g(n_2) - 1$ points.  The total number of points required is 
$f(n_1) + g(n_2) - 1$.

If all points of $Q_A$ lie to the left of $p$, then 
place $r_0$ at the bottommost point of $Q_A$ and observe that we now have the situation of $f$-draw-1 with $Q_R$ empty, and $f(n_1) + g(n_2) - 1$ points suffice. 
%
%let $q$ be the bottommost point of $Q_A$. 
%See Figure~\ref{fig:fdraw1-gdraw}(right). Place $r_0$ at $q$ and draw the edge $(p,r_0)$ down and left.  Start a ray upwards from $q$ and recursively draw $T_1$ in $f$-configuration in the subrectangle of $Q_A$ above $q$.  The leftward and downward rays at $r_0$ are free, so $T_2$ can be recursively drawn in $g$-configuration in $Q_B$.

This proves:
\begin{align}\tag{$g$}
g(n) \le f(n_{1}) + g(n_{2}).
\end{align}

We now describe a different $f$-drawing method that is more efficient than $f$-draw-1 above, and will be the key for our bound for binary trees.  

\smallskip\noindent{\bf $f$-draw-2.}  This method applies to an $f$-configuration.  %We will state the bound on the size of $Q$ below. 
We begin as in $f$-draw-1, though with the $f$-drawing and the $g$-drawing switched.
Partition $Q$ by a horizontal half-grid line $h$ and the ray from $p$ down to $h$ into 3 parts: $Q_B$, the rectangle below $h$; $Q_L$, the upper left rectangle; and $Q_R$, the upper right rectangle. Choose $h$ to be the highest half-grid line such that $Q_L$ or $Q_R$ has $g(n_1)$ points.
Without loss of generality, assume the former. We separate into two cases depending on the size of $Q_R$.

If $|Q_R| < g(n_{2,1})$ then we follow the $f$-draw-1 method.  Let $p_1$ be the bottommost point of $Q_L$.  Place $r_0$ at $p_1$, draw the edge $(p,r_0)$ down and left, recursively draw $T_1$ in $g$-configuration in $Q_L$ using leftward/upward rays from $r_0$, and recursively draw $T_2$ in $f$-configuration in $Q_B$ using a downward ray from $r_0$. This requires $g(n_1) + g(n_{2,1}) + f(n_2) -1$ points, where $g(n_{2,1})$ accounts for the wasted points in $Q_R$. 

If $|Q_R| \ge g(n_{2,1})$ then we make use of the points in $Q_R$ by drawing subtree $T_{2,1}$ there.
%This method works so long as $Q$ has at least $2g(n_1) + f(n_{2,2}) + n - 1$ points.
Let $p_1$ be the bottommost point of $Q_L$, and let $p_2, p_3, \ldots$ be the points of $Q_B$ below $p_1$ in decreasing $y$-order.  Let $k \ge 2$ be the smallest index such that either $k=n$ or point $p_{k+1}$ lies to the right of $p_{k}$. See Figure~\ref{fig:fdraw1-gdraw}(c). 

We have two subcases.  If $k= n$, then $p_1, \ldots, p_k$ form a monotone chain of length $n$, i.e., a diagonal point set in the terminology of Di Giacomo et al.~\cite{DiGiacomo}.  They showed that any tree of $n$ points can be embedded on a diagonal point set, so we simply
draw $T$ on these $n$ points. (Note that if this construction is used in the induction step, upward visibility is  needed for connecting $T$ to the rest of the tree, and this can be achieved.) 

Otherwise $k<n$.  
%\ournote{Check the accounting of points! Martin says: accounting should be good}
Place $r_0$ at point $p_{k}$ and $r_1$ at $p_{k+1}$. Draw the edge $(p,r_0)$ down and left, and the edge $(r_0, r_1)$ down and right.
Recursively draw $T_1$ in $g$-configuration in $Q_L$ using leftward/upward rays from $r_0$.  Draw $T_{2,2}$ in $f$-configuration in the rectangle below $r_1$ using a downward ray from $r_1$.  Draw $T_{2,1}$ in $g$-configuration in $Q_R$ using the rightward ray from $r_1$. 
Observe that if $r_1$ lies to the right of $p$ (i.e., below $Q_R$ rather than below $Q_L$) then the upward ray from $r_1$ is clear (as required for a $g$-drawing).
The number of points required is at most $2g(n_1) + n + f(n_{2,2}) - 1$.
This accounts for at most $g(n_1)$ points in $Q_R$, and at most $n$ points below $h$ and above $r_1$.

This proves:
\begin{align*}\tag{$f$-2}
f(n) \le \max \lbrace  
   g(n_1) + g(n_{2,1}) + f(n_2), \ 
   2g(n_1) + f(n_{2,2}) + n \rbrace. 
\end{align*}

%\subsection{Analysis} 

%In this section, we prove upper bounds for the number of points.
% TB: perfect was defined earlier, no need to re-define
%We first analyze
%the case of binary trees that are \emph{perfect}, i.e., all leaves are at %the same height and all non-leaf nodes have 2 children. 

\begin{theorem}
\label{thm:perfect-binary}
Any perfect binary tree with $n$ nodes has an L-shaped drawing on any point set of size $c \cdot n^{1.142}$ for some constant $c$. 
\end{theorem}
\begin{proof}
%We will use the methods $g$-draw and $f$-draw-2 from above, i.e.,~we will use recurrence relations ($g$) and ($f$-2).  
For perfect binary trees we have $n_1 = n_2 = \frac{1}{2}n$ and $n_{2,1} = n_{2,2} = \frac{1}{4}n$. 
%We use $\bar f, \bar g$ for the functions in this special case.
\newcommand{\myd}{\ensuremath{19.388}}
\newcommand{\myc}{\ensuremath{23.382}}
% 2^alpha = 2.206
%
%We show by induction on $n$ that $\bar g(n)\leq \beta cn^\alpha$ and $\bar f(n)\leq c n^\alpha$, for $\alpha = 1.142$, $\beta=1/(2^\alpha-1)\approx 0.8286$, and $c=\myc$. The computations can be found 
We solve the simultaneous recurrence relations 
for $f$ and $g$
in 
\ifarxiv
Appendix~\ref{app:perfect-binary} by induction.
\else
in the full paper.
\fi
\end{proof}

%\subsection{Analysis for General Binary Trees}

%In this section, give a bound for L-shaped drawings of general binary trees.

%Now we turn to general binary trees.

\newcommand{\myalpha}{\ensuremath{1.220}} %{\ensuremath{1.222}}  % alpha
% I changed this to 1.220 as per Therese's note later on.  Anna.
\newcommand{\myp}{\ensuremath{2.333}}  % 2^alpha
\newcommand{\mybeta}{\ensuremath{0.75}}   % 1/(2^alpha-1); this is beta in Timothy's notes
\newcommand{\myd}{\ensuremath{83.333}}
\newcommand{\myc}{\ensuremath{112}}%{\ensuremath{111.111}}

\begin{theorem}
\label{thm:general-binary}
Any binary tree has an L-shaped drawing on any point set of size $c \cdot n^{1.22}$ for some constant $c$. 
\end{theorem}
\begin{proof}
For $n_1\le 0.349n$, we use recursion ($f$-1).
For $n_{2,1}\le 0.082n$, we combine recursion
($f$-1') and ($f$-1) to obtain
$f(n) \leq 2g(n_{1}) + 2f(n_{2,1}) + g(n_{2,2})$.
For $n_{1} > 0.349n$ and $n_{2,1} > 0.082n$, we use 
recursion ($f$-2).
We solve the simultaneous recurrence relations for $f$ and $g$ in
%We claim that $ g(n)\leq \beta  cn^\alpha$ and $ f(n)\leq c n^\alpha$, for $\alpha = \myalpha$,
%$\beta=1/(2^\alpha-1)\approx 0.7522$, and $c=\myc$.
%The details can be found in 
\ifarxiv
Appendix~\ref{app:general-binary} by induction.
\else
the full paper.
\fi
\end{proof}

%%%%%%%%%%%%%%%%%%%%%%%%%%%%%%%%%%%%%%%%%%%%%%%%%%%%
\section{Trees of Maximum Degree 4}
\label{sec:degree-4}

In this section, we prove bounds on the number of points needed for L-shaped drawings of trees with maximum
degree $4$. We treat the trees as rooted and refer  to them as ternary trees. 
Given a ternary tree of $n$ nodes, let
$a_1$, $b_1$ and $r_1$ be the three children of the root $r_0$.
We use $T_v$ to denote the subtree rooted at a
node $v$, and $n_v$ to denote the number of nodes in $T_v$.
 We name the children of the root such that
$n_{a_1} \leq n_{b_1}\leq n_{r_1}$.   For $i\geq 2$, let $a_i,b_i,r_i$ be the three
children of $r_{i-1}$, named such that $n_{a_i}\leq n_{b_i}\leq n_{r_i}$.
See Figure~\ref{fig:tree-set-up}(b).

%Let $k\geq 2$ be the smallest number such that
%$n_{r_k} \leq 0.9 n_{k-1}$.

We will draw ternary trees using only the $f$-configuration as defined in Section~\ref{sec:degree-3} (see Figure~\ref{fig:tree-set-up}(c)).  
In this section (as opposed to the previous one) we define $f(n)$ to be minimum number such that any ternary tree of $n$ nodes can be drawn in $f$-configuration on any set of $f(n) - 1$ points. 

\remove{
We use the same $f$-configuration for drawings that was used in Section~\ref{sec:degree-3}.
% \iffalse
% \todo{This definition might well be used earlier for binary trees}
% We consider {\em $f$-type drawings} (see Figure~\ref{fig:fType}):
% The points reside within some rectangle $R$ that has
% no edges drawn within it.  There is a ray $\rho$ that hits
% $R$ vertically from above, and which comes from the parent of
% the subtree that we need to draw.
% \fi
We use $f(n)$ for the minimum number such that any set of $f(n)-1$ points
in the $f$-type setup can be used to draw any ternary tree with $n$ nodes.

%\begin{figure}[ht]
%%\hspace*{\fill}
%%\includegraphics[width=0.25\linewidth,page=1]{drawing_types.pdf}
%\hspace*{\fill}
%\includegraphics[width=0.7\linewidth,page=2]{deg4tree.pdf}
%\hspace*{\fill}
%\caption{$f$-type drawings.}
%\label{fig:fType}
%\label{fig:deg4}
%\end{figure}

Fix a tree with $n$ nodes that we want to draw.
Let $a_1$, $b_1$ and $r_1$ be the three children of root $r_0$.
In what follows, we use $T_v$ to denote the subtree rooted at a
node $v$, and $n_v=|T_v|$.
 We name the children of the root such that
$n_{a_1} \leq n_{b_1}\leq n_{r_0}$.   For $i\geq 2$, let $a_i,b_i,r_i$ be the three
children of $n_{i-1}$, named such that $n_{a_i}\leq n_{b_i}\leq n_{r_i}$.
Let $k\geq 2$ be the smallest number such that
$n_{r_k} \leq 0.9 n_{k-1}$.   
%See Figure \ref{fig:deg4}. 
See Figure~\ref{fig:tree-set-up}(b).
}

%\subsection{Drawing the Tree}
%%\subsection{Recursions}
%\label{sec:draw-ternary}

As in Section~\ref{sec:degree-3}, we will give various drawing methods, each of which gives a recurrence relation for $f(n)$.  
\ifarxiv
Then in Appendix~\ref{app:ternary} we will analyze the recurrence relations.
\fi
We begin with a re-description of Aichholzer et al.'s method~\cite{Aichholzer}.
%that behaves well when $n_{a_1} + n_{b_1}$ is relatively small.
%\ournote{From Anna.  Maybe I have this explanation backwards?    The main explanation below is that when $a_1$ and $b_1$ are small then we continue with $a_2, b_2$ etc. to level $k$.  However, this simple drawing method is used in the analysis for when $n_{b_1} \le .47n$.  TB: But that is small.  We want to avoid the case $(0,\frac{n}{2},\frac{n}{2})$, so $n_{b_1}\le .47n$ is (comparatively) small.  I've changed `large' to `small' here.}

\smallskip\noindent{\bf $f_4$-draw-1.} Extend the vertical ray from $p$
downward to a horizontal half-grid line $h$ determined as follows. 
Partition $Q$ by $h$ and the ray down to $h$ into 3 parts: $Q_B$, the rectangle below $h$; $Q_L$, the upper left rectangle; and $Q_R$, the upper right rectangle. 
Choose $h$ to be the highest half-grid line such that $Q_L$ or $Q_R$ has  $2f(n_{a_1}) + 2 f(n_{b_1})$ points.
Without loss of generality, assume the former. 
Partition $Q_L$ vertically into two rectangles $Q_{LL}$ and $Q_{LR}$ with 
atleast $f(n_{a_1})$ points on the left and
atleast $f(n_{b_1})$ points on the right respectively, with $Q_{LL}$ to the left of $Q_{LR}$.  Place $r_0$ at the bottommost point in $Q_{LR}$.
Extend a ray upward from $r_0$ and recursively draw $T_{b_1}$ on the remaining $f(n_{b_1})-1$ points in $Q_{LR}$.
Extend a ray to the left from $r_0$ and recursively draw $T_{a_1}$ on the $f(n_{a_1})$ points in $Q_{LL}$.
Finally, extend a ray downward from $r_0$ and recursively draw $T_{r_1}$ in $Q_B$.  See Figure~\ref{fig:f4draw}(a).
The number of points required is $2f(n_{a_1}) + 2 f(n_{b_1}) + f(n_{r_1}) - 1$,
so this proves:  
\begin{align}\tag{$f_4$-1}
f(n)\leq 2f(n_{a_1})+2f(n_{b_1}) + f(n_{r_1}).
\end{align}

For example, in the case when $T$ is perfect (with $n_{a_1}=n_{b_1}=n_{r_1}=\tfrac{n}{3}$),
the inequality ($f_4$-1)  
becomes $f(n)\leq 5f(n/3)$, which resolves to $O(n^{\log_3 5})$ and $\log_3 5 \approx 1.465$.  The critical case for this recursion, though, turns out to be when
$n_{a_1}=0$ and $n_{b_1}=n_{r_1}=\tfrac{n}{2}$, which gives
$f(n)\le 3f(n/2)$ and leads to Aichholzer et al.'s $O(n^{\log_3 2})$ result.

\begin{figure}[htb]
\centering
\includegraphics[width=\linewidth]{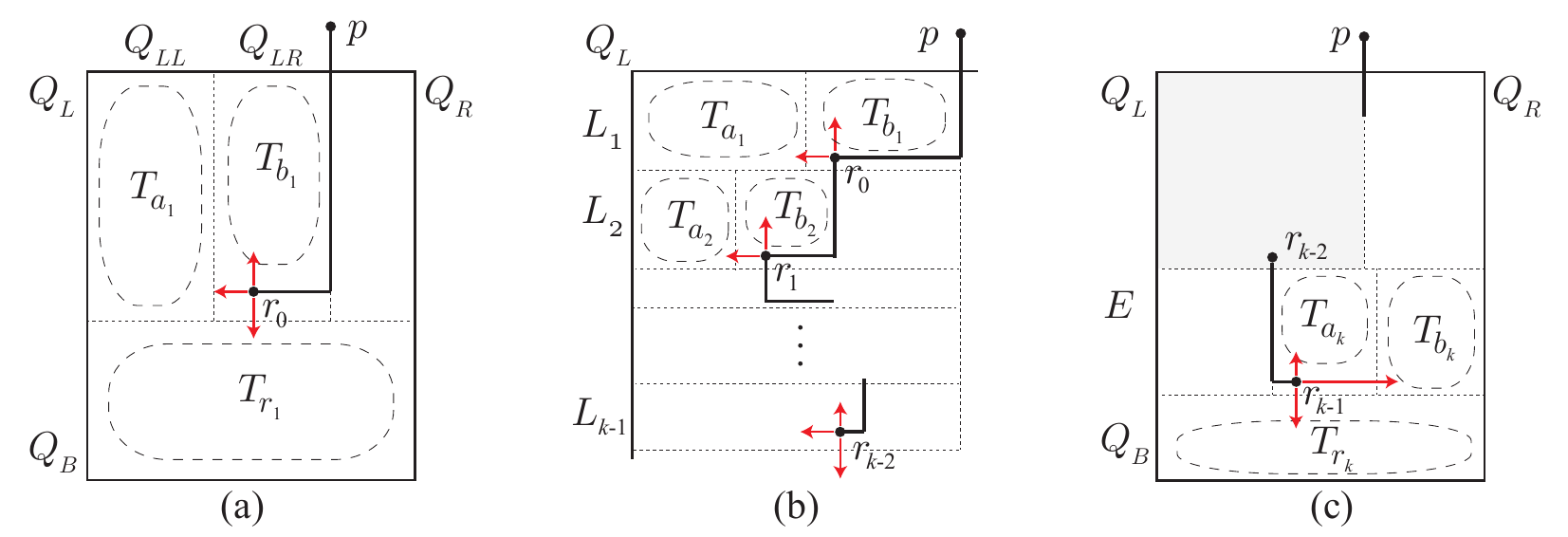}
\caption{(a) $f_4$-draw-1. (b) Drawing the ``small'' subtrees in $Q_L$.  (c) $f_4$-draw-2A.}
\label{fig:f4draw}
\end{figure}

\smallskip\noindent{\bf $f_4$-draw-2.} 
To improve their result, our idea again is to avoid wasting the points in $Q_R$, and use some of those points in the recursive drawings of subtrees at the next levels.
However, simply considering subtrees at the second level is not sufficient for an asymptotic improvement if $T_{a_2}$ and $T_{b_2}$ are too small. Thus, we consider a more complicated approach that 
%For finer control of the drawing, we will explicitly draw more of the subtrees, 
stops at the first level $k \ge 2$ where 
$n_{r_k} \leq 0.9 n_{r_{k-1}}$.
Note that for $i=2,\ldots,k-1$, we have $n_{r_i} > 0.9 n_{r_{i-1}}$ and $n_{a_i},n_{b_i}\le 0.1 n_{r_{i-1}}$, and so
$T_{a_i}$ and $T_{b_i}$ are ``small'' subtrees.
We apply the same idea as above to draw not just $T_{a_1}, T_{b_1}$ but also all the small subtrees $T_{a_i}$ and $T_{b_i}$, $i = 2, \ldots, k-1 $ in $Q_L$ (appropriately defined), and then consider a few cases for how to draw the remaining ``big'' subtrees $T_{a_k}, T_{b_k}$, and $T_{r_k}$, possibly using some points in $Q_R$.
The number of points we will need to reserve for drawing $T_{a_1},T_{b_1},\ldots, T_{a_{k-1}},T_{b_{k-1}}$ is 
\begin{align*}
    Y = f(n_{a_1}) + f(n_{b_1}) + \sum_{i=2}^{k-1} ( 2f(n_{a_i}) + 2f(n_{b_i})). 
\end{align*}
Extend the vertical ray from $p$ downward until $Q_L$ or $Q_R$ has $Y$ points.
Without loss of generality, assume the former.  

\smallskip\noindent{\bf Drawing the small subtrees.}
We draw nodes $r_i$ and subtrees $T_{a_i}$ and $T_{b_i}$, $i = 1, \ldots, k-1 $ in $Q_L$ as follows.
Split $Q_L$ horizontally into rectangles $L_1, \ldots, L_{k-1}$.
The plan is to draw $r_i,T_{a_i}$ and $T_{b_i}$ in $L_i$, in the same way that $T_{a_1}$ and $T_{b_1}$ were drawn in $f_4$-draw-1.
See Figure~\ref{fig:f4draw}(b).
Level $L_1$ is special because the vertical ray from $p$ is at the right boundary of $L_1$.  Thus, we require $f(n_{a_1}) + f(n_{b_1})$ points.
For levels $L_i$, $i = 2,  \ldots, k-1$ the vertical ray from $r_{i-2}$ may enter $L_i$ at any point, so we require 
$2f(n_{a_i}) + 2f(n_{b_i})$ points to follow the plan of $f_4$-draw-1, and the L-shaped edge from $r_{i-2}$ to $r_{i-1}$ may turn left or right.
The total number of points we need in all levels is $Y$, which is why we defined $Y$ as we did.
%\footnote{Our definition of $Y$ was over-generous; we could have used $2k-1$ fewer points and would still be able to draw all subtrees.  However, this does not seem to lead to an asymptotic improvement in the area.}. 

\smallskip\noindent{\bf Drawing the final three subtrees.}
It remains to draw $r_{k-1}$ and its three subtrees $T_{a_k}$, $T_{b_k}$, and $T_{r_k}$.
We will draw $T_{r_k}$ on the bottommost $f(n_{r_k}) - 1$ points of $Q$.  Call this rectangle $Q_B$.
Let $E$ be the ``equatorial zone'' that lies between $Q_L,Q_R$ above and $Q_B$ below.  See Figure~\ref{fig:f4draw}(c).
If we are lucky, then not too many points are wasted in $Q_R$. Let $Z \le Y$ be the number of points in $Q_R$.%  We have $Z \le Y$.

\begin{description} 
    \item[Case A:] %\ournote{Note that this was previously Case 4.}  
    $Z < f(n_{b_k})$.  In this case we draw $r_{k-1}, T_{a_k}$ and $T_{b_k}$ in $E$ as in $f_4$-draw-1.  See Figure~\ref{fig:f4draw}(c).
    For this, we need $2f(n_{a_k}) + 2 f(n_{b_k})$ points in $E$.
    The total number of points required in this case is $Y + Z + 2f(n_{a_k}) + 2 f(n_{b_k}) + f(n_{r_k}) - 1$,
so this proves:  
\begin{align*}\tag{$f_4$-2A}
f(n) & \leq Y {+} Z {+} 2f(n_{a_k}) {+} 2 f(n_{b_k}) {+} f(n_{r_k})\\
%    & \leq Y + 2f(n_{a_k}) + 3 f(n_{b_k}) + f(n_{r_k}) \\
    & = f(n_{a_1}) {+} f(n_{b_1}) {+} \sum_{i=2}^{k-1} \left( 2f(n_{a_i}) {+} 2f(n_{b_i})\right) {+} 2f(n_{a_k}) {+} 3 f(n_{b_k}) {+} f(n_{r_k}).
\end{align*}
\end{description}

We must now deal with the unlucky case when $Z \ge f(n_{b_k})$.  
We will require $3f(n_{a_k}) + f(n_{b_k})$ points in $E$. We sum up the total number of points below, but first we describe how to complete the drawing in $E$.
Partition $E$ into three regions: $E_L$, $E_M$, and $E_R$, where $E_L$ is the region to the left of $r_{k-2}$, $E_R$ is the region to the right of $p$, and $E_M$ is the region between them. See Figure~\ref{fig:f4drawB}.
Observe that either $| E_L| \ge f(n_{a_k}) + f(n_{b_k})$, or $|E_M| \ge f(n_{a_k})$, or $|E_R| > f(n_{a_k})$.
We show how to draw $r_{k-1}, T_{a_k}$ and $T_{b_k}$ in each of these 3 cases.

\begin{figure}[htb]
\centering
\includegraphics[width=\linewidth]{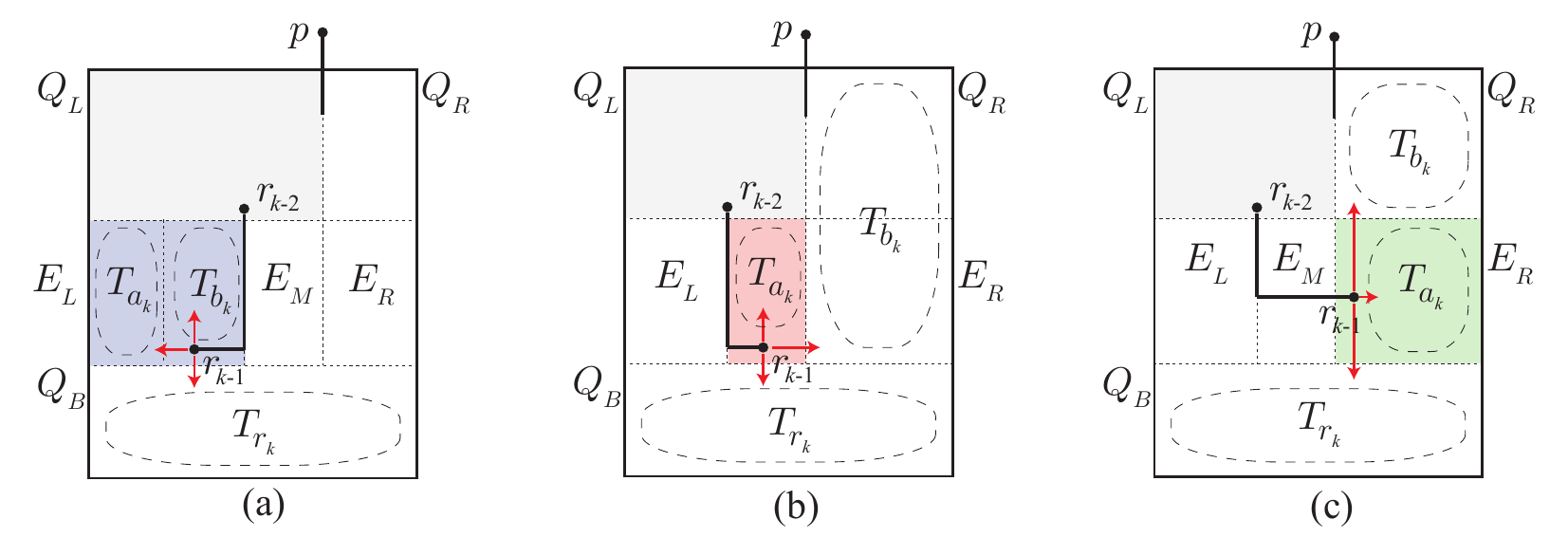}
\caption{The drawings for $f_4$-draw-2B: (a) Case B1, with $E_L$ in blue;  (b) Case B2, with $E_M$ in red;  (c) Case B3, with $E_R$ in green. }
\label{fig:f4drawB}
\end{figure}

\begin{description}
    \item[Case B1:] $|E_L| \geq f({n_{a_k}})+f(n_{b_k})$.  In this case we draw $r_{k-1}, T_{a_k}$ and $T_{b_k}$ in $E_L$ as in $f_4$-draw-1.  See Figure~\ref{fig:f4drawB}(a). Since $E_L$ is to the left of the ray down from $r_{k-2}$, we have sufficiently many points.

    \item[Case B2:] $|E_M| \ge f(n_{a_k})$.  In this case we place $r_{k-1}$ at the lowest point of $E_M$, draw $T_{a_k}$ above it in $E_M$, and $T_{b_k}$ to its right in $Q_R\cup E_R$. See Figure~\ref{fig:f4drawB}(b). Since $|Q_R|=Z \ge f(n_{b_k})$, we have enough points to do this.

    \item[Case B3:] $|E_R| > f(n_{a_k})$.   In this case we place $r_{k-1}$ at the leftmost point of $E_R$, draw $T_{a_k}$ to its right in $E_R$ and $T_{b_k}$ above it in $Q_R$.   See Figure~\ref{fig:f4drawB}(c). Again, there are sufficiently many points.
\end{description}
The total number of points required in each of these three cases is $Y + Z + 3f(n_{a_k}) + f(n_{b_k}) + f(n_{r_k}) - 1$, and $Y \le Z$ which yields:
\begin{align*}\tag{$f_4$-2B}
f(n) & \leq Y + Z + 3f(n_{a_k}) + f(n_{b_k}) + f(n_{r_k})\\
    & =  2f(n_{a_1}) + 2f(n_{b_1}) + \sum_{i=2}^{k-1} ( 4f(n_{a_i}) + 4f(n_{b_i})) + 3f(n_{a_k}) + f(n_{b_k}) + f(n_{r_k}).
\end{align*}

The bound on $f(n)$ obtained from $f_4$-draw-2 is the maximum of ($f_4$-2A) and ($f_4$-2B).
%\ournote{Maybe this should be made clearer with a formula.}

%%%%%%%%%%%%%%%%%%%%%%%%%%%%%%%%%%%%%%%%%%%%%%%%%%%%%%%%%%%%%%%%%
%\subsection{Analysis for Ternary Trees}
%\label{sec:analysis-ternary}

%We first consider the case when $T$ is perfect, 
%i.e., $n_{a_1}=n_{b_1}=n_{r_1}=n/3$, etc.
%In this case the method $f_4$-draw-1, which was already used by Aichholzer et al.~\cite{Aichholzer},
%gives a better bound than was stated in their paper.  In particular, 
%the inequality ($f_4$-1)  
%becomes $f(n)\leq 5f(n/3)$, which resolves to $O(n^{\log_5 3})$ and $\log_5 3 \approx 1.465$.  
%We do not improve this, but we improve the best known bound for the case of general ternary trees to the following:

\renewcommand{\myc}{\ensuremath{2}}

\begin{theorem}
\label{thm:ternary}
	Any ternary tree with $n$ nodes has an L-shaped drawing on any point set of size $2n^{1.55}$. 
	%points in general position 
%	for some constant $c$.
\end{theorem}
\begin{proof}
For $n_{b_1} \leq 0.47n$, we use recursion ($f_4$-1).
Otherwise, we use ($f_4$-2A) or ($f_4$-2B) and take
the larger of the two bounds.
We solve the recurrence relation for $f$ in 
\ifarxiv
Appendix~\ref{app:ternary} by induction.
\else
the full paper.
\fi
\remove{
We will show by induction on $n$ that $f(n) \leq cn^\alpha$ for $\alpha=1.55$ and $c=\myc$. 
The bound holds for $n=1$ since one point is enough.  Now assume that the bound holds for all values  $<n$.
We split the induction step into two cases based on the size of $T_{b_1}$.
The algorithm uses the best-possible recursion, which means that it suffices
to show that the bound holds for one of the recursive formulas for $f$.
For details, see 
\ifarxiv
Appendix~\ref{app:ternary}.
\else
the full paper.
\fi
}
\remove{  %%% all this is in the appendix.  I hoped to keep a bit here, but then there are Conclusions and Acknowledgements and it does not fit.  AL.
\begin{itemize}
\item Case 1: \textbf{$n_{b_1} \leq 0.47n$}.
By ($f_4$-1) and the induction hypothesis, we know
$
    f(n) \leq 2c (n_{a_1})^\alpha + 2c(n_{b_1})^\alpha + c(n_{r_1})^\alpha. 
$
Since
the trivariate function $F(n_{a_1},n_{b_1},n_{r_1})=2 (n_{a_1})^\alpha + 2(n_{b_1})^\alpha + (n_{r_1})^\alpha$ is convex, it suffices to check whether
the bound holds for the extreme points of the convex region
$\{(n_{a_1},n_{b_1},n_{r_1})\in [0,n]^3 : n_{a_1}+n_{b_1}+n_{r_1} \le n,\ n_{a_1}\le n_{b_1}\le n_{r_1},\ n_{b_1}\le 0.47n\}$.  In this case, the extreme points (excluding the origin) are
$(0, 0, n)$, $(\tfrac{n}{3},\tfrac{n}{3},\tfrac{n}{3})$,  and $(0, 0.47n, 0.53n)$.  
Since
\begin{align*}
2(0) + 2(0) + (n^\alpha) & \leq n^\alpha \\
2(\frac{n}{3})^\alpha + 2(\frac{n}{3})^\alpha + (\frac{n}{3})^\alpha & < 0.911n^\alpha \\	
2(0) + 2(0.47n)^\alpha + (0.53n)^\alpha & < 0.995n^\alpha, 
\end{align*}
we have $f(n)\leq cn^\alpha$ in this case.

\item Case 2: \textbf{$n_{b_1} > 0.47n$}.  See 
\ifarxiv
Appendix~\ref{app:ternary}.
\else
the full paper.
\fi
\end{itemize}
}
\end{proof}

\section{Conclusions}
We have made slight improvements on the exponent $t$ in the bounds that $c \cdot n^t$ points always suffice for drawing trees of maximum degree 4, or 3, with L-shaped edges.     
Improving the bounds to, e.g.,~$O(n \log n)$ will require a breakthrough. In the other direction, there is still no counterexample to the possibility that $n$ points suffice.

We introduced the problem of drawing ordered trees with L-shaped edges, where many questions remain open.  For example: Do $c \cdot n$ points suffice for drawing ordered caterpillars?  Can our isolated example be expanded to prove that $n$ points are not sufficient in general?

\smallskip\noindent{\bf Acknowledgments}
We thank Jeffrey Shallit for investigating the alternating sequences discussed in Section~\ref{sec:caterpillars}.
This work was done as part of a Problem Session in the Algorithms and Complexity group at the  University of Waterloo.  We thank the other participants for helpful discussions.

\newpage
\bibliographystyle{splncs03}
\bibliography{ortho-trees}

\ifarxiv
\appendix\clearpage
\section{Proof of Lemma~\ref{claim:14-caterpillar}}
\label{appendix:14-caterpillar}

\begin{figure}[t]
\hspace*{\fill}
\includegraphics[width=0.40\textwidth]{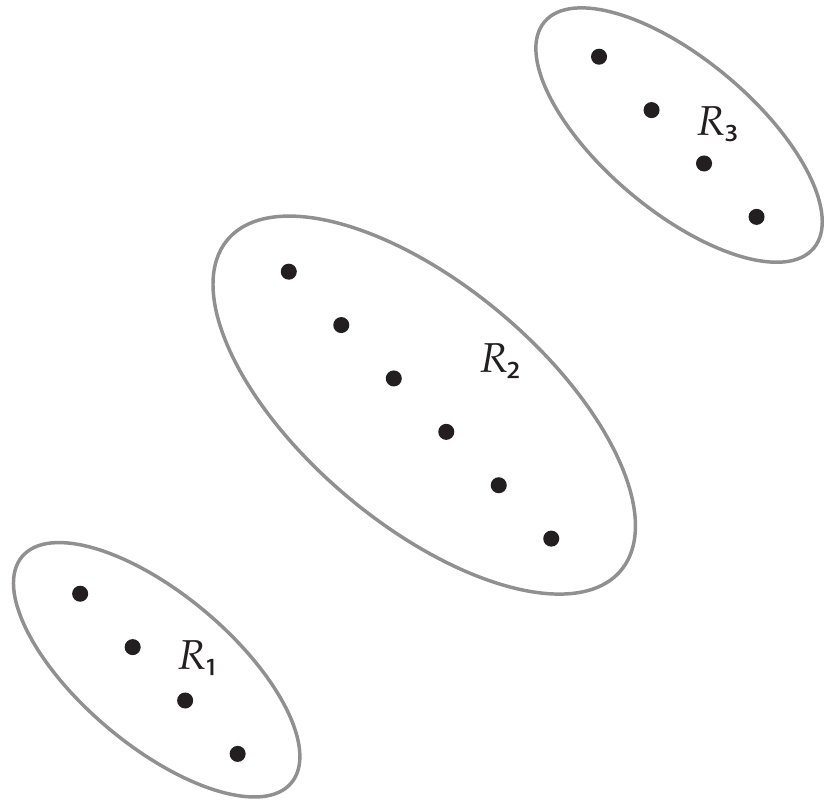}
\hspace*{\fill}
\includegraphics[width=0.40\textwidth]{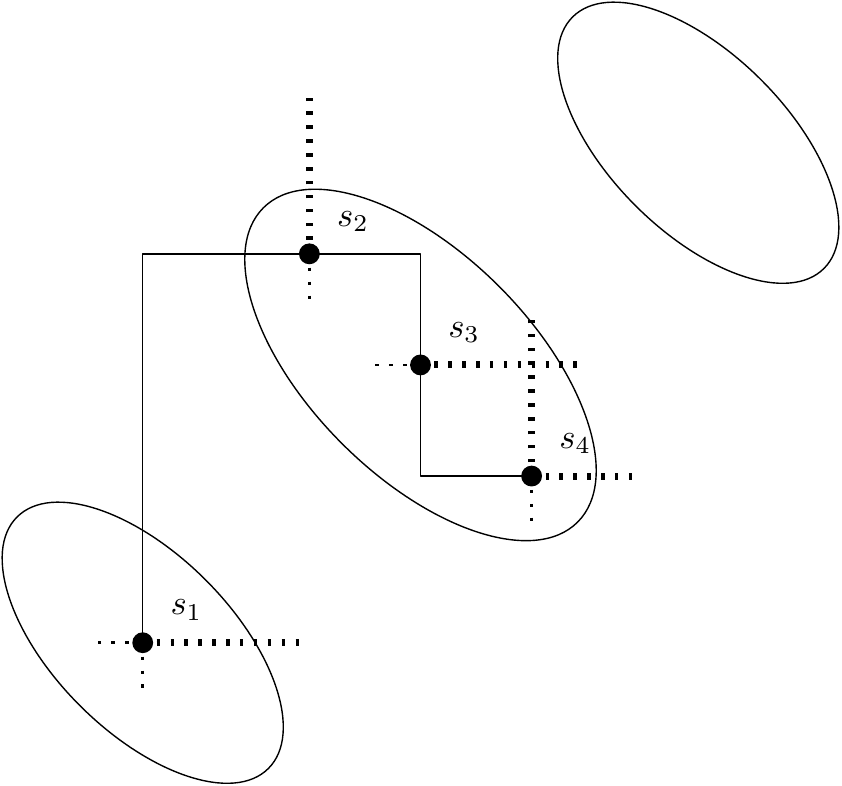}
\hspace*{\fill}
\caption{Point set $P_{14}$ of size $14$.}
\label{fig:catpointset}
\end{figure}

Let us denote the three groups of points by $R_1$, $R_2$ and $R_3$ as shown in Fig.~\ref{fig:catpointset}. Let $s_1,\dots,s_4$ be the {\em spine-vertices} (the vertices of degree 4 of $C_{14}$), in order along the spine. We prove Lemma~\ref{claim:14-caterpillar} using the following two claims:

\begin{claim}
\label{claim:not-consecutive}
    Let $r$ and $t$ be two points in $R_i$ for some $i \in \{1,2,3\}$ that are both assigned spine-vertices of $C_{14}$ such that $r$ is to the left of and above $t$.  Then $r$ and $t$ are not consecutive in the $x$-order of points of $R_i$.
\end{claim}
\begin{proof}
The bottom ray of $r$ must be used since spine-vertices have degree 4.  If $r$ and $t$ are consecutive, then no point lies between them either in $x$-direction or in $y$-direction, which means that the bottom ray of $r$ either connects to $t$ or goes beyond the $y$-coordinate of $t$.  Likewise the left ray of $t$ either connects to $r$ or goes beyond the $x$-coordinate of $r$, but the latter is impossible since then the bottom ray of $r$ would intersect the left ray of $t$.  So $(r,t)$ exists and is routed along the bottom of $r$ and the left of $t$.  Repeating the argument with the right ray of $r$ and the top ray of $t$ gives that $(r,t)$ is a double edge, a contradiction.
\end{proof}

\remove{
\begin{claim}
	The $R_2$ region of point set $P_{14}$ has at most two vertices of degree 4 of $C_{14}$ as shown in Figure~\ref{fig:caterpillar}(a).
\end{claim}
\begin{proof} We prove the claim by contradiction.
	Let there be a drawing $\Gamma$ in which three points $p$, $q$, and $r$ in $R_2$ (listed in increasing $x$ order) represent vertices of degree $4$. By Claim~\ref{claim:not-consecutive}, there must be another point in $R_2$ between every two of them. Thus, a vertical line $L$ through the middle point $q$ must have exactly two points of $R_2$ on one side of it. Without loss of generality, let us say that it is on the right side. Let $s$ be the fourth point assigned to a degree 4 vertex of $C_{14}$. Based on the position of $s$ we analyze three cases:
% 	\iffalse
% 	In $\Gamma$ they must be drawn in one of the two ways as shown in Figure \ref{fig:3point}.
	
% 	\begin{figure}
% 		\centering
% 		\includegraphics[width=0.3\linewidth,page=2]{caterpillar_pointset}
% 		\caption{Two ways of drawing three points $p,q$ and $r$ in $R_2$.}
% 		\label{fig:3point}
% 	\end{figure}
	
% 	w.l.o.g let assume it is one that is shown in the left side of the Figure \ref{fig:3point}. 
	
% 	\fi
	%Let $L$ be the vertical line 

	\begin{itemize}
		\item Case 1: Vertices represented by $p$, $q$ and $r$ are consecutive in the spine of $C_{14}$ and $s$ is in $R_3$. In this case, we need at least five points on the right side of the $L$ (at least 1 for the leaf of $q$, 1 for the leaf of $r$ and 3 for the leaves of $s$) but we have only $4$ free points available right of $L$.
		
% 		%

% 		it cannot be adjacent to $q$ because then $q$ cannot be connected to two of its leaves therefore, $s$ must be adjacent to $p$ and $r$ (i.e., connecting q to a point in $R_{1}$ - as its leaf node - would block $s$ from being connected to either $p$ or $r$). This now implies that $r$ will have one of its leaf nodes in $R_{1}$, but this is not possible given that $s$ must be connected to both $p$ and  $r$.

		\item Case 2: This case is symmetric to Case 1 by rotating the point set by $90^\circ$ counter-clockwise and flipping along the $x$-axis. After the rotation and flip, the number of points right of $L$ in $R_2$ remains the same, and the point $s$ is in the $R_3$ group.
% Vertices represented by $p$, $q$ and $r$ are consecutive in the spine of $C_{14}$ and $s$ is in $R_{1}$. 
% In this case, we have the position of $p$ and $q$ fixed, and we know that $s$ is the only point assigned to a degree 4 vertex in $R_1$. We also know that the points in $R_3$ must represent the leaves connected to vertices in $R_2$ () 

% As the drawing needs to be straight-through in $p$ and $q$, we know that $s$ is connected to $p$ via top ray, or to $r$ via right ray. 

% 		In this case, we have 5 free points available on the right side of $L$. We call this set $S_{R}$. Due to the order of edges in $C_{14}$:
% 		\begin{itemize}
% 			\item $q$ must have 1 leaf node in $S_{R}$ and 1 in $S_L$.
% 			\item $p$ can have at most 1 leaf node in $S_{R}$.
% 			\item $r$ can have at most 2 leaf nodes in $S_{R}$.
% 			\item $s$ can have no leaf nodes in $S_{R}$. (This is because connecting $s$ - which is in $R_1$ - to a point in $R_3$ would block $s$ from connecting it to at least one of its neighbor.)
% 		\end{itemize}
% 		But there are only 4 points which are used in $S_R$ which implies that, in the drawing $\Gamma$, we have at least 1 point in $S_{R}$ that is not mapped to any vertex.
		\item Case 3: Vertices represented by $p$, $q$ and $r$ are not consecutive in the spine of $C_{14}$.  In this case, $s$ must be adjacent to two of $p$, $q$ and $r$. Assume that $s$ is in $R_3$ (the case when $s$ is in $R_1$ is analogous). As $s$ is strictly above and right of all $p$, $q$ and $r$, it can be connected to them only using its left and bottom ray. Thus, the spine would need to make a bend at $s$, and the two leaves connected to $s$ would need to be placed
		on the same side of the spine, which violates the order of edges at $s$.
	\end{itemize}
\end{proof}
}
	
\begin{claim}
\label{claim:notr1}
    There is at most one
    spine-vertex in $R_1$, and it is either $s_1$ or $s_4$.
\end{claim}
\begin{proof}
No vertex of degree 4 can be on the leftmost or bottommost point of $R_1$, since the left/bottom ray from it could not be used.
The remaining two points cannot be both assigned to spine-vertices by Claim~\ref{claim:not-consecutive}, so at most one spine-vertex belongs to $R_1$.

Now assume for contradiction that $s_2\in R_1$.  By the order-constraints the edges $(s_1,s_2)$ and $(s_2,s_3)$ 
are both incident to $s_2$ vertically or both horizontally.
%are either both horizontal or both vertical.  
Say they are both vertical, 
which means that one of $s_1,s_3$ is lower down than $s_2$,
%farther left than $s_2$, 
and therefore also in $R_1$.  This contradicts that only one spine-vertex belongs to $R_1$.  Similarly we obtain a contradiction if both $(s_1,s_2)$ and $(s_2,s_3)$ are horizontal, so $s_2\not\in R_1$ and similarly $s_3\not\in R_1$.
\end{proof}

Similarly neither $s_2$ nor $s_3$ are in $R_3$.  We know that
at most three spine-vertices are in $R_2$ by Claim~\ref{claim:not-consecutive}, so at least one spine-vertex is in $R_1\cup R_3$, and it must be $s_1$ or $s_4$.  Say $s_1\in R_1$ (all other cases are symmetric). 

We may assume that edge $(s_1,s_2)$ leaves $s_1$ vertically, the other case is the same after a diagonal flip of the point set.  Since $s_2\in R_2$, edge $(s_1,s_2)$ arrives at $s_2$ horizontally and from the left.  So $(s_2,s_3)$ leaves $s_3$ horizontally to the right, and hence must go downward to reach $s_3$, because $s_3\in R_2$.  So $(s_2,s_3)$ reaches $s_3$ from the top, which means that $(s_3,s_4)$ leaves $s_3$ from the bottom and $s_4\not\in R_3$.  But also $s_4\not\in R_1$ since $s_1\in R_1$ and as argued above only one spine-vertex belongs to $R_1$.  So $s_4\in R_2$ as well.

This means that the four points in $R_3$ are all used for leaves.
There are only five leaves for which the corresponding edges could possibly reach $R_3$:  the top ray from $s_2$, the right ray from $s_3$, the top and right ray from $s_4$, and the right ray from $s_1$.  If both the right ray from $s_3$ and the top ray from $s_4$ go towards points in $R_3$ then they will intersect, contradicting planarity. Thus, not both of these can be using leaves in $R_3$, which means that the right ray from $s_1$ must go to $R_3$.   But then the right ray from $s_1$ blocks any of the left/bottom rays from $s_2,s_3,s_4$ from reaching $R_1$.  In consequence, only the rays from $R_1$ can reach leaves in $R_1$, leaving at least one point in $R_1$ unused.  Contradiction, so $C_{14}$ has no embedding on $P_{14}$.

\remove{

Analogously to Claim~\ref{claim:notr1}, one can argue that there is at most 1 vertex of degree 4 in $R_3$. As cannot $R_2$ contain more than $2$ vertices of degree $4$, we know that it contains precisely $2$ such vertices and $R_1$ and $R_3$ contain one $1$ each.

The both the degree 4 vertex represented by a point $r_1$ in $R_1$ and the degree 4 vertex represented by a point $r_4$ in $R_3$ must be connected to three leaves as it is impossible to maintain a straight-through drawing in $r_1$ and $r_4$. At the same time, it is not possible to maintain straight-through drawing at both the points that represent the degree $4$ vertices in $R_2$ while leading the spine from $r_1$ to $r_4$. 
}

\section{A Connection between Straight-through Drawings of Paths and  Alternating Runs in Sequences}
\label{app:runs}

The problem of finding monotone straight-through paths is related to the following problem about alternating runs in a sequence.
Given a sequence $s_1$, $s_2, \ldots , s_n$, whose elements are a permutation of $1, \ldots, n$ 
find a maximum size set of indices $i_1 \le i_2 \cdots \le  i_k$
such that the subsequence  $s_{i_1}, s_{i_2}, \ldots, s_{i_k}$
is \emph{3-good}, meaning that it
consists of alternating \emph{runs} of length at least 3.
A \emph{run} is an increasing sequence or a decreasing sequence.
For example, the subsequence   $1, 3, 6, 4, 2, 10, 12, 13$ is 3-good since its three runs, 
$1, 3, 6$ and $6, 4, 2$ and $2, 10, 12, 13$ have lengths 3, 3, and 4 respectively.
The subsequence  $1, 4, 3$, $2, 10, 12, 13$
is not 3-good because its first run, $1, 4$ is too short.

The monotone straight-through drawing problem differs from the alternating runs problem in that
the straight-through drawing problem requires all runs to be of odd length, but, on the other hand, tolerates shorter runs at the beginning and the end.
For the alternating runs problem, it seems that a sequence of length $2n$ always has a 3-good subsequence of length $n$,
\emph{except} for the sequence $5, 2, 6, 3, 1, 7, 4$---this sequence has length 7, but its longest 3-good subsequence has length 3.  

%\ournote{Should we say something about experiments?} 
%Jeff has tried all $n$ up to some small value, plus millions of random examples for larger $n$.  Experimental constant: %1.35. (inverse .74).

\section{Analysis for Perfect Binary Trees}
\label{app:perfect-binary}

\begin{proof}[of Theorem~\ref{thm:perfect-binary}]
%We will use the methods $g$-draw and $f$-draw-2 from above, i.e.,~we will use recurrence relations ($g$) and ($f$-2).  
For perfect binary trees we have $n_1 = n_2 = \frac{1}{2}n$ and $n_{2,1} = n_{2,2} = \frac{1}{4}n$. 
We use $\bar f, \bar g$ for the functions in this special case.
%\newcommand{\myd}{\ensuremath{19.388}}  -- these are defined in main text
%\newcommand{\myc}{\ensuremath{23.382}}
% 2^alpha = 2.206
%
We claim that $\bar g(n)\leq \beta cn^\alpha$ and $\bar f(n)\leq c n^\alpha$, for $\alpha = 1.142$, $\beta=1/(2^\alpha-1)\approx 0.8286$, and $c=24$. %\myc$.
%\ournote{TC: I think using $\beta$ is better than $d$, since it doesn't depend on $c$, the exact value of which is unimportant.  TB: I agree, that looks nicer.  TC: Also, for some reason, my calculations suggest a different $c\approx 22.883...$.  TB: Entirely possible that some rounding-error caused that.  If you're sure about your calculations then please change.}
%$d=\myd$, and $c=(2^\alpha-1)d
%%\approx 1.206\cdot \myd  
%\approx \myc$.
Clearly we have $\bar g(1)\leq \beta c$ and $\bar{f}(1)\leq c$ since one point is enough to draw the tree.
Now assume that the bounds hold for all values $<n$.  We apply $(g)$ and have
$$\bar g(n)\leq \bar f(\tfrac{1}{2}n)+ \bar g(\tfrac{1}{2}n) \leq c (\tfrac{1}{2}n)^\alpha + \beta c\cdot (\tfrac{1}{2}n)^\alpha  
=\tfrac{1+\beta}{2^\alpha} cn^\alpha =\beta c n^\alpha$$
as desired (since $\beta$ is chosen so that $1+\beta=2^\alpha\beta$).  For function $f$, the algorithm uses the best of the recursions, which means that
it is no worse than recursion ($f$-2), and we have
\begin{eqnarray*}
\bar f(n) & \leq & \max\{ 2 \bar g(\tfrac{1}{2}n) + \bar f(\tfrac{1}{4}n) + n,\, 
    \bar g(\tfrac{1}{2}n)+ \bar g(\tfrac{1}{4}n) + \bar f(\tfrac{1}{4}n) \} \\
    & \leq & \max\{ 2\beta c(\tfrac{1}{2}n)^\alpha + c(\tfrac{1}{4}n)^\alpha  
    + n^\alpha , \, \beta c (\tfrac{1}{2}n)^\alpha  + \beta c(\tfrac{1}{4}n)^\alpha  + c(\tfrac{1}{4}n)^\alpha \}  
	\leq cn^\alpha
%\\
%    & \approx & \max\{ \frac{\myd}{1.103}+ \frac{\myc}{4.866 } + 1,
%    	\frac{\myd}{2.206} + \frac{\myd}{4.866} + \frac{\myc}{2.206} \} n^\alpha \\
%    & \approx & \max\{17.577 + 4.805 + 1 ,
%    8.789 + 3.984 + 10.599 \} n^\alpha  \leq \myc\cdot n^\alpha
\end{eqnarray*}
since (with our choice of $\alpha,\beta,c$) we have
\begin{align*}
 2\beta c(\tfrac{1}{2})^\alpha + c(\tfrac{1}{4})^\alpha  + 1 < 0.957c + 1  & < c
 %\approx  0.9563c + 1\leq c
%\tfrac{2}{2.206}\beta c+\tfrac{1}{2.206}c+1 \approx 17.577 + 4.805 + 1 \leq \myc= c
\quad\text{and}\\
\beta c (\tfrac{1}{2})^\alpha + \beta c(\tfrac{1}{4})^\alpha + c(\tfrac{1}{4})^\alpha 
\leq 0.999c & < c. %0.9987c & < c.
%\approx \tfrac{1}{2.206} d + \tfrac{1}{4.866}d + \tfrac{1}{2.206} c
%\approx 8.789 + 3.984 + 10.599 \leq \myc = c.
\end{align*}
%(In the proofs to come, we will leave the latter computations to the reader.)
\end{proof}

(A more careful analysis shows that the exponent in Theorem~\ref{thm:perfect-binary} approaches $\log_2 x$ where $x$ is the real root of the cubic polynomial $x^3-2x^2-1=0$.)

\section{Analysis for General Binary Trees}
\label{app:general-binary}

\begin{proof}[of Theorem~\ref{thm:general-binary}]
We claim that $ g(n)\leq \beta  cn^\alpha$ and $ f(n)\leq c n^\alpha$, for $\alpha = \myalpha$,
$\beta=1/(2^\alpha-1)\approx 0.7522$, and $c=112$.
%\ournote{TB: If we want 3 digits, use $\alpha=1.220$.}
Clearly we have $ g(1)\leq \beta c$ and ${f}(1)\leq c$ since one point is enough to draw the tree.
Now assume that the bound holds for all values $<n$ and consider the recursive formulas.  
As in the previous proof we have $g(n)\leq \beta c n^\alpha$ since $\beta=1/(2^\alpha-1)$.

As for $f$, the algorithm always uses the best-possible recursion, so it suffices (for
various cases of how nodes are distributed in the subtrees) to argue that for one of
the recursions we have $f(n)\leq cn^\alpha$. 

\begin{itemize}
\itemsep -2pt
\item{Case 1: $n_{1} \leq 0.349n$.} We use recursion ($f$-1), i.e., $f(n)\leq 2g(n_1)+f(n_2)$.
Applying induction, hence $f(n)\leq 2\beta cn_1^\alpha +cn_2^\alpha$.
%= (\frac{2}{2^\alpha-1}+1)c (n_1^\alpha+n_2^\alpha) \approx 1.666c (n_1^\alpha+n_2^\alpha)$.
For this and the other cases,
since the bivariate function $F(n_1,n_2) = 2\beta cn_1^\alpha + cn_2^\alpha$ is convex
over the domain $\{(n_1,n_2)\in [n]^2 : n_1+n_2\le n,\ n_1\le 0.349n\}$,
%$f$ is a convex function in range $[0,\infty)$, 
it suffices to 
%substitute values and 
check that the bound 
holds at the extreme points of the domain (see \cite[Lemma 10]{Scheucher-thesis}).
%\ournote{TC: the convexity argument needs to be expressed in terms of a bivariate function...}
The extreme points for $(n_{1},n_{2})$ in this case (ignoring the origin) are $(0,n)$ and $(0.349n, 0.651n)$.   
For all of them we have $f(n)\leq cn^\alpha$ since
\begin{align*}
2c(0)^\alpha + \beta c(n)^\alpha = \beta c n^\alpha & < 0.753cn^\alpha 
%\leq cn^\alpha 
\quad \text{and}\\
2c(0.349n) ^\alpha + \beta c(0.651n)^\alpha  &
%\approx (2 \cdot 0.278 + \mybeta \cdot 0.589) cn^\alpha 
< 0.9993 cn^\alpha.% \leq cn^\alpha.
\end{align*}

\item Case 2: $n_{2,1} \leq 0.082n$. 
We have $f(n) \leq 2g(n_{1}) + f(n_2)$ by ($f$-1') and $f(n_2) \leq 2f(n_{2,1}) + g(n_{2,2})$ by ($f$-1),
so $f(n) \leq 2g(n_{1}) + 2f(n_{2,1}) + g(n_{2,2})
\leq 2\beta c n_{1}^\alpha + 2cn_{2,1}^\alpha + \beta cn_{2,2}^\alpha$. The extreme points for
$(n_{1},n_{2,1},n_{2,2})$ are $(0,0,n)$, $(0,0.082n,0.918n)$,
$(\frac{1}{2}n,0,\frac{1}{2}n)$, and $(\frac{1}{2}n,0.082n,0.418n)$. 
For all of them we have $f(n)\leq cn^\alpha$ since
\begin{eqnarray*}
2\beta c(0)^\alpha + 2c(0)^\alpha + \beta c(n)^\alpha & < & 0.753 c n^\alpha %\leq cn^\alpha
\\
2\beta c(0)^\alpha + 2c(0.082 n)^\alpha + \beta c(0.918n)^\alpha & < & 0.773 cn^\alpha %\leq cn^\alpha 
\\ 
2\beta c(\tfrac{1}{2}n)^\alpha + 2c(0)^\alpha + \beta c(\tfrac{1}{2}n)^\alpha & < & 0.969 cn^\alpha \\
%\leq cn^\alpha\\
2\beta c(\tfrac{1}{2}n)^\alpha + 2c(0.082n)^\alpha + \beta c(0.418n)^\alpha & < & 0.99991 cn^\alpha.
%\leq cn^\alpha.
\end{eqnarray*}

\item Case 3: $n_{1} > 0.349n$ and $n_{2,1} > 0.082n$. 
Using recursion ($f$-2), we know that 
\begin{align*}
f(n) & \le \max \lbrace  g(n_1) + g(n_{2,1}) + f(n_2),\,  2g(n_1) + f(n_{2,2}) + n \rbrace \\
& \leq \max \lbrace \beta cn_1^\alpha + \beta cn_{2,1}^\alpha + cn_2,\,  2\beta cn_1^\alpha + cn_{2,2}^\alpha + n^\alpha \rbrace.
\end{align*}
The extreme points for 
$(n_{1},n_{2,1},n_{2})$ are $(\frac{1}{2}n,\frac{1}{4}n,\frac{1}{2}n)$ and $(0.349n,0.3255n,\allowbreak 0.651n)$ 
and the extreme points for
$(n_{1},n_{2,2})$ are $(\frac{1}{2}n,0.418n)$ and $(0.349n,\allowbreak 0.569n)$. 
For all of them we have $f(n)\leq cn^\alpha$ since
%\todo{TB: I have put all numbers to 3 digits precision.  However, I have not re-checked the math; there might be some (tiny) rounding errors here.}
\begin{eqnarray*}
\beta c (\tfrac{1}{2}n)^\alpha + \beta c(\tfrac{1}{4}n)^\alpha  + c(\tfrac{1}{2}n)^\alpha& < & 0.891 cn^\alpha\\
%\leq cn^\alpha\\
\beta c (0.349n)^\alpha + \beta c(0.322n)^\alpha + c(0.651n)^\alpha & < & 0.992 cn^\alpha\\%\leq cn^\alpha \\
2\beta c (\tfrac{1}{2} n)^\alpha + c(0.418n)^\alpha +n^\alpha & < & 
(0.991c+1)n^\alpha < cn^\alpha\\
2\beta c (0.349n)^\alpha + c(0.569n)^\alpha +n^\alpha & < & (0.920c+1)n^\alpha < cn^\alpha
\end{eqnarray*}
where the inequalities hold since we chose $c$ such that $0.991c+1 < c$. 
\end{itemize}
\end{proof}

\section{Analysis for Ternary Trees}
\label{app:ternary}

\begin{proof}[of Theorem~\ref{thm:ternary}]
We will show by induction on $n$ that $f(n) \leq cn^\alpha$ for $\alpha=1.55$ and $c=\myc$. 
The bound holds for $n=1$ since one point is enough.  Now assume that the bound holds for all values  $<n$.
We split the induction step into two cases based on the size of $T_{b_1}$.
The algorithm uses the best-possible recursion, which means that it suffices
to show that the bound holds for one of the recursive formulas for $f$.
%that for one of the recursive formulas for $f$ the bound holds.

\begin{itemize}
\item Case 1: \textbf{$n_{b_1} \leq 0.47n$}.
By ($f_4$-1) and the induction hypothesis, we know
$
    f(n) \leq 2c (n_{a_1})^\alpha + 2c(n_{b_1})^\alpha + c(n_{r_1})^\alpha. 
$
Since
the trivariate function $F(n_{a_1},n_{b_1},n_{r_1})=2 (n_{a_1})^\alpha + 2(n_{b_1})^\alpha + (n_{r_1})^\alpha$ is convex, it suffices to check whether
the bound holds for the extreme points of the convex region
$\{(n_{a_1},n_{b_1},n_{r_1})\in [0,n]^3 : n_{a_1}+n_{b_1}+n_{r_1} \le n,\ n_{a_1}\le n_{b_1}\le n_{r_1},\ n_{b_1}\le 0.47n\}$.  In this case, the extreme points (excluding the origin) are
$(0, 0, n)$, $(\tfrac{1}{3}n,\tfrac{1}{3}n,\tfrac{1}{3}n)$,  and $(0, 0.47n, 0.53n)$.  
Since
\begin{align*}
2(0) + 2(0) + (n^\alpha) & \leq n^\alpha \\
2(\tfrac{1}{3}n)^\alpha + 2(\tfrac{1}{3}n)^\alpha + (\tfrac{1}{3}n)^\alpha & < 0.911n^\alpha \\	
2(0) + 2(0.47n)^\alpha + (0.53n)^\alpha & < 0.995n^\alpha, 
\end{align*}
we have $f(n)\leq cn^\alpha$ in this case.

\item Case 2: \textbf{$n_{b_1} > 0.47n$}.
 We know that  $n_{r_1} < 0.53n$ and therefore $n_{b_k} \le n_{r_{k-1}}/2 < 0.265n$. 
 By ($f_4$-2A) and ($f_4$-2B) and the induction hypothesis,
 \begin{align*}
     f(n) \leq \max\left\{ \begin{array}{l}
     c(n_{a_1})^\alpha + c(n_{b_1})^\alpha + \sum_{i=2}^{k-1} ( 2c(n_{a_i})^\alpha + 2c(n_{b_i})^\alpha) + {}\\ 2c(n_{a_k})^\alpha + 3c (n_{b_k})^\alpha + c(n_{r_k})^\alpha,\\[3pt]
     2c(n_{a_1})^\alpha + 2c(n_{b_1})^\alpha + \sum_{i=2}^{k-1} ( 4c(n_{a_i})^\alpha + 4c(n_{b_i})^\alpha) +{}\\ + 3c(n_{a_k})^\alpha + c(n_{b_k})^\alpha + c(n_{r_k})^\alpha.
     \end{array} \right.
 \end{align*}
 
 The second term in the maximum dominates the first
%.  To see this, it suffices to check that $2c(n_{b_1})^\alpha + c(n_{b_k})^\alpha > c(n_{b_1})^\alpha + 3c(n_{b_k})^\alpha$, i.e., $2(n_{b_k})^\alpha < n_{b_1}^\alpha$, and the inequality follows since $2(0.265n)^\alpha < (0.47n)^\alpha$ for $\alpha=1.55$.  
since $2(n_{b_k})^\alpha \leq 2(0.265n)^\alpha < (0.47n)^\alpha \leq (n_{b_1})^\alpha$ for our choice of $\alpha=1.55$.  
\remove{
 The induction would follow if we can show:
 \begin{align*}\tag{$\ast$}
     2(n_{a_1})^\alpha + 2(n_{b_1})^\alpha + \sum_{i=2}^{k-1} ( 4(n_{a_i})^\alpha + 4(n_{b_i})^\alpha) + {}\\[-2pt] \ \ \ \ \ \ \ \ \ \ \ \ \ \ 3(n_{a_k})^\alpha + (n_{b_k})^\alpha + (n_{r_k})^\alpha\ <\ n^\alpha.
 \end{align*}
}
To show that the second term is at most $cn^\alpha$, we use three intermediate claims and show
\begin{align*}
& 2(n_{a_1})^\alpha {+} 2(n_{b_1})^\alpha {+} \textstyle{\sum_{i=2}^{k-1}} ( 4(n_{a_i})^\alpha {+} 4(n_{b_i})^\alpha) {+} 3(n_{a_k})^\alpha {+} (n_{b_k})^\alpha {+} (n_{r_k})^\alpha\ \\
& \leq 2(n_{a_1})^\alpha {+} 2(n_{b_1})^\alpha {+} \textstyle{\sum_{i=2}^{k-1}} ( 4(n_{a_i})^\alpha {+} 4(n_{b_i})^\alpha) {+} 0.92(n_{r_{k-1}})^\alpha \quad\text{(by Claim \ref{claim1})} \\
& \leq 2(n_{a_1})^\alpha {+} 2(n_{b_1})^\alpha {+} 0.92(n_{r_{1}})^\alpha \quad\text{(by Claim \ref{claim2} for $j=k{-}1,\dots,2$)} \\
& \leq n^\alpha \quad\text{(by Claim \ref{claim3}).} 
\end{align*}
 
% We prove ($\ast$) in three steps:
The three claims are proved as follows:
 
\begin{claim}\label{claim1}
 $3(n_{a_k})^\alpha+(n_{b_k})^\alpha +(n_{r_k})^\alpha < 0.92(n_{r_{k-1}})^\alpha$.
\label{f_k}
\end{claim}

We can check that the claim holds by calculating the values
for the extreme points of the 
region defined by our constraints, viz.,
$\{ (n_{a_k}, n_{b_k}, n_{r_k})\in [0,n]^3:\allowbreak n_{a_k}{+}n_{b_k} {+} n_{r_k} \le n_{r_{k-1}},\
n_{a_k}\le n_{b_k}\le n_{r_k}\le 0.9n_{r_{k-1}}\}$.  These are
the points 
	$(0,  \frac{n_{r_{k-1}}}{2},  \frac{n_{r_{k-1}}}{2}),
	(0, 0.1n_{r_{k-1}}, \allowbreak 0.9n_{r_{k-1}}),
	(0.05n_{r_{k-1}}, 0.05n_{r_{k-1}}, \allowbreak 0.9n_{r_{k-1}})$,
	and 
	$(\tfrac{1}{3}n_{r_{k-1}}, \tfrac{1}{3}n_{r_{k-1}}, \tfrac{1}{3}n_{r_{k-1}})$, 
and we verify:
\begin{align*}
		 0 + \big(\tfrac{1}{2}n_{r_{k-1}}\big)^\alpha+ \big(\tfrac{1}{2}n_{r_{k-1}}\big)^\alpha < 0.685n_{r_{k-1}}^\alpha \\
		  0 + (0.1n_{r_{k-1}})^\alpha+ (0.9n_{r_{k-1}})^\alpha < 0.878n_{r_{k-1}}^\alpha\\
		  3(0.05n_{r_{k-1}})^\alpha + (0.05n_{r_{k-1}})^\alpha+ (0.9n_{r_{k-1}})^\alpha < 0.888n_{r_{k-1}}^\alpha  \\
		 3\big(\tfrac{1}{3}n_{r_{k-1}}\big)^\alpha + \big(\tfrac{1}{3}n_{r_{k-1}}\big)^\alpha+ \big(\tfrac{1}{3}n_{r_{k-1}}\big)^\alpha < 0.911n_{r_{k-1}}^\alpha.
\end{align*}

%Now we express the bound for all higher layers of the tree:

\begin{claim}\label{claim2}
$4(n_{a_{j}})^\alpha+4(n_{b_{j}})^\alpha+0.92(n_{r_{j}})^\alpha < 0.92(n_{r_{j-1}})^\alpha$ for  $2 \leq j \leq k-1$.
\label{f_0} 
\end{claim}

%We prove the claim inductively moving up in the tree. At the bottom, for
%$j=k-1$, the bound holds by Claim~\ref{f_k} as $f(n_{r_{k-1}}) < 0.92(n_{r_{k-1}})^\alpha$. \ournote{I do not see this in Claim 1 at all. Martin}. \ournote{This follows from the fact that $n_{a_k} {+} n_{b_k} {+} n_{r_k} =n_{r_{k-1}}$}. 
We can check that the claim holds by calculating the values
for the extreme points of the region
$\{ (n_{a_j}, n_{b_j}, n_{r_j})\in [0,n]^3: n_{a_j}{+}n_{b_j} {+} n_{r_j} \le n_{r_{j-1}},\
n_{a_k}\le n_{b_k}\le n_{r_k}\le 0.9n_{r_{j-1}}\}$, specifically, the points
	$(0,  \tfrac{n_{r_{j-1}}}{2},  \tfrac{n_{r_{j-1}}}{2})$,
	$(0, 0.1n_{r_{j-1}}, \allowbreak 0.9n_{r_{j-1}})$,
	$(0.05n_{r_{j-1}}, 0.05n_{r_{j-1}}, 0.9n_{r_{j-1}})$,
	and 
	$(\tfrac{n_{r_{j-1}}}{3}, \tfrac{n_{r_{j-1}}}{3}, \tfrac{n_{r_{j-1}}}{3})$:
\begin{align*}
    0 + 0 + 0.92(n_{r_{k-2}})^\alpha &< 0.92(n_{r_{k-2}})^\alpha \\
    0 + 4(0.1n_{r_{k-2}})^\alpha+ 0.92(0.9n_{r_{k-2}})^\alpha &< 0.895(n_{r_{k-2}})^\alpha\\
	 4(0.05n_{r_{k-2}})^\alpha + 4(0.05n_{r_{k-2}})^\alpha+ 0.92(0.9n_{r_{k-2}})^\alpha &< 0.859(n_{r_{k-2}})^\alpha.
\end{align*}

%Assuming that the bound holds for all $j$ with $k-1 \geq j \geq m$,
%we use the same extreme points and check:

%\begin{align*}
%  0 + 0 + 0.92(n_{r_{m-2}})^\alpha & < 0.92(n_{r_{m-2}})^\alpha \\
% 0 + 4(0.1n_{r_{m-2}})^\alpha+ 0.92(0.9n_{r_{m-2}})^\alpha & < 0.895(n_{r_{m-2}})^\alpha\\
% 4(0.05n_{r_{m-2}})^\alpha + 4(0.05n_{r_{m-2}}^\alpha+ 0.92(0.9n_{r_{m-2}})^\alpha & < 0.859(n_{r_{m-2}})^\alpha \\
%\end{align*}

%Finally, we can state the bound on the highest level:
\begin{claim}\label{claim3}
$2(n_{a_1})^\alpha+2(n_{b_1})^\alpha+0.92(n_{r_1})^\alpha \leq n^\alpha$.
\end{claim}

%By Claim~\ref{f_0}, we know that $f(n_{r_1}) < 0.92(n_{r_1})^\alpha$.
%\ournote{Martin: Again, I fail to see this.} Checking the extreme points $(n_{a_1}, n_{b_1}, n_{r_1})$ with values $(0, \tfrac{1}{2}, \frac{1}{2})$
%and $(0,0,1)$ we verify that

We can check that the claim holds by calculating the values
for the extreme points of the region
$\{ (n_{a_1}, n_{b_1}, n_{r_1})\in [0,n]^3: n_{a_1}+n_{b_1} + n_{r_1} \le n,\
n_{a_1}\le n_{b_1}\le n_{r_1}, n_{b_1}>0.47n\}$, specifically, the points
	$(0.08n,  0.47n, 0.47n)$,
	$(0,  \tfrac{1}{2}n,  \tfrac{1}{2}n)$ and
	$(0,  0, n)$:
\begin{align*}
    	  2(0.08n)^\alpha + 2(0.47n)^\alpha + 0.92(0.47n)^\alpha &< 0.946n^\alpha \\ 
    	  0 + 2(\tfrac{1}{2}n)^\alpha + 0.92(\tfrac{1}{2}n)^\alpha &< 0.998n^\alpha \\ 
  0 + 0 + 0.92(n)^\alpha &< 0.92n^\alpha. 
\end{align*}
\end{itemize}
This finishes the proof of Theorem~\ref{thm:ternary}.
\end{proof}

\fi
\end{document}